\documentclass[11pt]{article}
\usepackage[a4paper]{geometry}
\usepackage[utf8]{inputenc}
\usepackage[T1]{fontenc}
\usepackage{lmodern}

\usepackage{tikz-cd}
\usepackage{amsmath, amsthm, amssymb, amsfonts, mathtools}
\usepackage{braket}
\usepackage{xcolor}
\usepackage{hyperref}
\hypersetup{
    colorlinks=true,
    linkcolor=blue,
    citecolor=blue,
    urlcolor=blue
}
\usepackage[style=alphabetic, sorting=nty]{biblatex} 
\usepackage{todonotes} 
\usepackage{graphicx}
\usepackage{caption}   
\usepackage{subcaption}  
\usepackage{authblk}
\usepackage[ruled,vlined,linesnumbered]{algorithm2e}

\SetAlgoInsideSkip{smallskip}
\SetAlgoNlRelativeSize{-1}
\DontPrintSemicolon
\usepackage{combelow}

\newtheorem{theorem}{Theorem}[section]
\newtheorem{lemma}[theorem]{Lemma}
\newtheorem{proposition}[theorem]{Proposition}
\newtheorem{corollary}[theorem]{Corollary}

\newtheorem{definition}[theorem]{Definition}

\newcommand{\HomI}[2]{\operatorname{Hom}_1(#1,#2)}
\newcommand{\im}{\operatorname{im}}

\def\london{Quantinuum, Partnership House, Carlisle Place, London, SW1P 1BX, UK}
\def\broomfield{Quantinuum, Broomfield, CO 80021, USA}

\title{Automated logical Clifford gadgets for heterogeneous architectures via chain maps}

\author[1]{Asmae Benhemou}
\affil[1]{\london}
\author[2]{Noah Berthusen}
\affil[2]{\broomfield}
\date{\today}

\begin{document}

\maketitle

\begin{abstract}
    Transversal CNOTs are ubiquitous for entangling logical qubits of identical CSS codes pairwise. For distinct codes, the options are much more limited, and are typically known only for structurally related code families. We introduce an automated framework for synthesising inter-code logical CNOT circuits between arbitrary CSS codes using chain maps. Given a prescribed bipartite logical CNOT network between these codes, our method constructs the affine space of chain maps realising the desired logical action, and then searches this space for shallow and sparse physical circuit candidates. We benchmark this method on a range of heterogeneous CSS code pairs, recovering known transversal constructions, and finding new low-depth solutions, including distance-preserving and partially distance-preserving examples, which we demonstrate can be promoted to the full code distance using additional flag measurements. We discuss applications to code switching, magic-state injection, Pauli product measurements, and operations on concatenated codes, where bespoke chain maps offer favourable spacetime tradeoffs for logical interfaces tailored to heterogeneous architectures. Finally, we show how our framework straightforwardly extends to targeted logical CZ gates.
\end{abstract}

\tableofcontents

\section{Introduction} 
\label{sec:introduction}

The earliest proposals for large-scale fault-tolerant quantum computing relied on homogeneous architectures, in which a single quantum error correcting code (QECC) is used to perform all required logical operations~\cite{Knill_2005, Litinski_2019}. Now, protocols~\cite{anderson2014, Heu_en_2025} and computational schemes~\cite{Sullivan_2024, gidney2024magic, tan2025, yoder2025, webster2026, cain2026} are considering heterogeneous architectures leveraging several distinct QECCs that are tailored for a specific purpose, whether that be memory, Clifford computation, or magic state preparation. Realising the benefits of such architectures requires the ability to efficiently interconnect the many potentially disparate QECC components. 

The most natural inter-code logical operation is the transversal CNOT gate: when both code blocks are CSS codes, a physical CNOT gate between matched physical qubits provides a logical CNOT between matched logical qubits. However, for heterogeneous architectures in which logical CNOT gates between distinct CSS codes are necessary, finding such operations requires a more structured approach. A key insight, developed in recent work on homomorphic CNOTs~\cite{huang2023homomorphic, xu2025fast}, is that such inter-code couplings are naturally described by chain maps between the chain complexes representing the two codes. A degree-1 chain map $\gamma_1: C_1^B \to C_1^A$ specifies a pattern of physical CNOT gates from code $A$ to code $B$ that, by virtue of the chain map conditions, automatically preserves the product codespace and induces a well-defined linear action on the logical operators of the code.

Despite this conceptual clarity, existing constructions of inter-code gadgets have largely been developed only for pairings that are closely related structurally, namely between 2D and 3D codes, and codes and their punctured subcodes~\cite{anderson2014, Heu_en_2025, li2025, tan2025}. Other proposals also typically rely on symmetry-preserving operations between identical code blocks~\cite{grassl2013leveraging}. A systematic and automated approach for finding low-overhead gadgets for arbitrary pairs of CSS codes and prescribed logical actions has been absent. 

In this work, we present a linear-algebraic framework for constructing and searching the space $\mathrm{Hom}_1(B, A)$ of all degree-1 chain maps between two CSS codes, and characterise how targeting a specific logical CSS action restricts this space to an affine subspace. We then introduce an optimisation-based procedure which identifies low-depth, sparse maps within that subspace. The resulting algorithm is fully automated; that is, given parity-check matrices for two CSS codes and a desired linear map describing a bipartite logical CNOT network between the two code blocks, it constructs the chain-map hom-space, imposes the appropriate subspace restrictions, and searches for physically efficient gadgets using a CP-SAT solver with a sparse basis heuristic.

Numerically, we find that the affine freedom left after logical targeting is large enough to reduce circuit overhead for the code sizes considered here. Across the examples studied, the optimiser finds either distance-preserving circuits or partially fault-tolerant circuits whose distance can be restored using additional flagging while preserving favourable spacetime overheads. We simulated examples of such circuits against idling-memory baselines and verified that the resulting gadgets have per-logical-channel performance on the same order of magnitude as the corresponding joint idle-memory baseline and the worse transversal-CNOT benchmark of the two codes, sometimes within a small constant factor.

We demonstrate the utility of this approach across several practically relevant settings. As the main application, we construct bespoke chain maps to facilitate code switching between a range of CSS code pairs, recovering known transversal constructions, and identifying new low-depth solutions, including both fault-tolerant and partially fault-tolerant circuits that offer favourable spacetime tradeoffs for near-term architectures. Notably, we identify maps which can implement magic state injection in instances that would have otherwise required expensive universal adapters~\cite{swaroop2026}. More broadly, bespoke chain maps enable a range of logical gadgets beyond code switching; for instance, the efficient realisation of logical FANOUT/IN gates allows the construction of useful gadgets for concatenated codes and large entangled states. The examples we present in this work are by no means exhaustive, and we argue that the presented tool provides a flexible framework for designing logical operations tailored to both heterogeneous and homogeneous fault-tolerant architectures. 

The remainder of this work is organised as follows. In Sec.~\ref{sec:preliminaries}, we introduce the chain complex formalism for CSS codes and define chain maps between chain complexes. In Sec.~\ref{sec:build-hom-space} we describe how to construct the degree-1 hom-space $\mathrm{Hom}_1(B, A)$ explicitly as the kernel of a binary matrix, and how imposing a logical action restricts this space to an affine subspace. In Sec.~\ref{sec:optimisation}, we cast the search for low-overhead representatives as a discrete optimisation problem over this affine family. Sec.~\ref{sec:applications} demonstrates the framework on several applications including magic state injection, code switching, Pauli product measurements, and operations on concatenated codes, comparing against existing constructions where available. Finally, we discuss some future directions and open problems in Sec.~\ref{sec:discussion}.

\section{Homomorphic CNOTs from chain maps} 
\label{sec:preliminaries} 

\subsection{CSS codes as chain complexes}
\label{ssec:codes_bg}

An $[[n,k,d]]$ CSS code can be represented by a length-$3$ chain complex over $\mathbb{F}_2$
\begin{equation}
C_2 \xrightarrow{\;\partial_2\;} C_1 \xrightarrow{\;\partial_1\;} C_0,
\end{equation}
where $C_0 \cong \mathbb{F}_2^{r_X}$, $C_1 \cong \mathbb{F}_2^{n}$, and $C_2 \cong \mathbb{F}_2^{r_Z}$ are finite-dimensional vector spaces with fixed computational bases, respectively identified with the $X$-check space, data-qubit space, and $Z$-check space. The boundary maps satisfy $\partial_1\partial_2=0$, and encode the $Z$ and $X$ stabilizer generators, so that $\partial_2 = H_Z{}^{\mathsf T}$ and $\partial_1 = H_X$ define the parity-check matrices of the code. The CSS orthogonality condition $H_XH_Z^{\mathsf T} = 0$ follows. The $Z$ logical operators are identified with homology classes in 
\begin{equation}
    H_1 = \ker \partial_1 / \mathrm{im}\,\partial_2.
    \label{eq:homology}
\end{equation}
Let us henceforth choose binary logical operator representatives with $ L_X L_Z^{\mathsf T}=I_k$, such that $\bar X_i$ anticommutes with $\bar Z_i$ and commutes with
$\bar Z_j$ for $i\neq j$.

\subsection{Chain maps between CSS codes}
\label{ssec:chain-maps-css}

Consider two CSS codes $A$ and $B$ as presented in Sec.~\ref{ssec:codes_bg}, with vector spaces 
\[C_{0}^{A(B)} \cong \mathbb{F}_2^{r_X^{A(B)}}, \qquad C_{1}^{A(B)} \cong \mathbb{F}_2^{n_{A(B)}}, \qquad C_2^{A(B)} \cong \mathbb{F}_2^{r_Z^{A(B)}},\]
represented by chain complexes $A_{\bullet}$ and $B_{\bullet}$ respectively, and boundary maps $\partial_2^{A(B)}$ and $\partial_1^{A(B)}$. The homology and cohomology follow accordingly. A \emph{chain map} $\boldsymbol{\gamma} : B_\bullet \xrightarrow{} A_\bullet$ is a triple of linear maps that define a morphism of chain complexes, i.e., $\boldsymbol{\gamma}  = (\gamma_2, \gamma_1, \gamma_0)$ such that
\begin{equation}
\gamma_2 : C_2^B \to C_2^A, \qquad
\gamma_1 : C_1^B \to C_1^A, \qquad
\gamma_0 : C_0^B \to C_0^A,
\end{equation}
that satisfy the following relations
\begin{equation}
\partial_2^A \gamma_2 = \gamma_1 \partial_2^B,
\qquad
\partial_1^A \gamma_1 = \gamma_0 \partial_1^B.
\label{eq:chain-map-conditions}
\end{equation}
These relations give rise to the following commutative diagram.
\begin{equation}
\begin{tikzcd}
  C_2^B \arrow[r, "\partial_2^B"] \arrow[d, "\gamma_2"'] 
    & C_1^B \arrow[r, "\partial_1^B"] \arrow[d, "\gamma_1"'] 
    & C_0^B \arrow[d, "\gamma_0"] \\
  C_2^A \arrow[r, "\partial_2^A"'] 
    & C_1^A \arrow[r, "\partial_1^A"'] 
    & C_0^A
\end{tikzcd}
\end{equation}

Expressing the boundary operators explicitly in terms of the parity-check matrices, the conditions~\eqref{eq:chain-map-conditions} give rise to the linear equations
\begin{equation}
(H_Z^A){}^{\mathsf T} \gamma_2 = \gamma_1 (H_Z^B){}^{\mathsf T},
\qquad
H_X^A \gamma_1 = \gamma_0 H_X^B,
\end{equation}
where $\gamma_0 \in \mathbb{F}_2^{r_X^A \times r_X^B}$, $\gamma_1 \in \mathbb{F}_2^{n_A \times n_B}$ and $\gamma_2 \in \mathbb{F}_2^{r_Z^A \times r_Z^B}$ act on stabilizer and qubit vector spaces. These equations define a commutative diagram that ensures that $X(Z)$-stabilizers are mapped to $X(Z)$-stabilizers, and that $\gamma_1$ represents a well-defined map on the homology, cf. Lemma~\ref{lemma:induced_op}. 
\begin{lemma}
    Any chain map $\boldsymbol{\gamma}:B_\bullet \xrightarrow{} A_\bullet$ induces a well-defined linear map
    \begin{equation}
        \gamma_\star : H_1(B) \xrightarrow{} H_1(A), \qquad [v] \xrightarrow{} [\gamma_1 v].
    \end{equation}
    \label{lemma:induced_op}
\end{lemma}
\begin{proof}
    Let $v \in \ker \partial_1^B$, then since $\boldsymbol{\gamma}$ is a chain map, it follows that $\partial_1^A(\gamma_1 v) = \gamma_0\partial_1^B v = 0$, so $\gamma_1 v \in \ker \partial_1^A$ (i.e. cycles are mapped to cycles). Next, if $v$ and $v'$ represent the same homology class in $H_1(B)$, then $v' = v + \partial_2^B w$ for some $w \in C_2^{B}$. Then, by \eqref{eq:chain-map-conditions} it follows that $\gamma_1v' = \gamma_1v + \partial_2^A\gamma_2w$. Hence, $\gamma_1v'$ and $\gamma_1v$ differ by an element of $\im\partial_2^A$, and therefore belong to the same equivalence class $[\gamma_1v] = [\gamma_1v']$ in $H_1(A)$ (i.e. boundaries are mapped to boundaries).  
\end{proof}
The set of all chain maps $\boldsymbol{\gamma} : B_\bullet \xrightarrow{} A_\bullet$ is given by the space $\mathrm{Hom}_{\mathrm{Ch}}(B,A)$, which is a linear subspace of the vector space
\[
\mathbb{F}_2^{r_Z^A \times r_Z^B} \oplus \mathbb{F}_2^{n_A \times n_B}  \oplus \mathbb{F}_2^{r_X^A \times r_X^B}.
\]
In particular, we are interested in $\operatorname{Hom}_1(B, A)$, that is 
the set of middle maps defined as the projection of the full solution space onto the linear subspace  
\begin{equation}
    \mathrm{Hom}_1(B, A) := \{ \gamma_1 \in \mathbb{F}_2^{n_A \times n_B} \mid \exists \gamma_2, \gamma_0 \text{ s.t } (\gamma_2, \gamma_1, \gamma_0) \in \mathrm{Hom}_{\mathrm{Ch}}(B, A)\}.
    \label{eq:hom_1_space}
\end{equation}
This space consists precisely of the degree-1 maps $\gamma_1$ that admit a completion to the full chain map. The following proposition gives an explicit characterisation of this condition purely in terms of $\gamma_1$. 
\begin{proposition}
A linear map $\gamma_1 \in \mathbb{F}_2^{n_A\times n_B}$ lies in $\operatorname{Hom}_1(B,A)$ if and only if
\begin{align}
\gamma_1\!\left(\im\partial_2^B\right)
&\subseteq \im\partial_2^A,
\label{eq:hom1-cond1}\\
\gamma_1\!\left(\ker\partial_1^B\right)
&\subseteq \ker\partial_1^A.
\label{eq:hom1-cond2}
\end{align}
\label{prop:hom1}
\end{proposition}
\begin{proof}
    Suppose that $\gamma_1 \in \HomI{B}{A}$. Then, by definition, there exist linear maps $\gamma_2: C_2^B\to C_2^A$ and $\gamma_0:C_0^B\to C_0^A$ that satisfy the chain map conditions in Eqs.~\eqref{eq:chain-map-conditions}. Let $u \in \im \partial_2^B$ and $u = \partial_2^B w$ for some $w \in C_2^B$. Then, 
    \[
    \gamma_1 u = \gamma_1\partial_2^B w = \partial_2^A\gamma_2 w \in \im\partial_2^A.
    \]
    Thus, $\gamma_1(\im\partial_2^B) \subseteq \im\partial_2^A$. 
    If $v\in \ker\partial_1^B$, then 
    \[
    \partial_1^A(\gamma_1 v) = \gamma_0\partial_1^Bv = 0,
    \]
    so that $\gamma_1 v \in \ker\partial_1^A$, and therefore $\gamma_1(\ker\partial_1^B) \subseteq \ker\partial_1^A$.
    Conversely, assume that conditions~\eqref{eq:hom1-cond1} and~\eqref{eq:hom1-cond2} hold. The first inclusion implies that $\im(\gamma_1\partial_2^B) \subseteq \im\partial_2^A$, so there exists a map $\gamma_2:C_2^B\to C_2^A$ such that \[ \partial_2^A\gamma_2=\gamma_1\partial_2^B. \] For the second condition, let
    \[
    \gamma_0(\partial_1^B x):=\partial_1^A\gamma_1x,
    \]
    for some $x \in C_1^B$. If $\partial_1^B x=\partial_1^B x'$ for some $x' \in C_1^B$, then
    $x+x'\in\ker\partial_1^B$, and hence
    $\gamma_1(x+x')\in\ker\partial_1^A$ by~\eqref{eq:hom1-cond2}. Therefore
    \[
    \partial_1^A\gamma_1x=\partial_1^A\gamma_1x',
    \]
    so this defines a linear map $\gamma_0 : \im\partial_1^B\to C_0^A$. Choosing any linear extension of $\gamma_0$ from $\im\partial_1^B$ to $C_0^B$ yields
    \[
    \gamma_0\partial_1^B=\partial_1^A\gamma_1.
    \]
    Thus, $\boldsymbol{\gamma} = (\gamma_2, \gamma_1, \gamma_0)$ satisfies the chain map conditions, and therefore $\gamma_1 \in \HomI{B}{A}$.
\end{proof}

\section{Linear-algebraic synthesis of chain maps}
\label{sec:build-hom-space}

In this section, we provide an explicit construction of the space of chain maps between two CSS chain complexes. This reduces the problem of enumerating chain maps to the computation of the kernel of an explicit binary matrix. We then derive a direct kernel description of the projected space $\operatorname{Hom}_1(B,A)$, and explain how additional requirements give rise to linear or affine subspaces of this hom-space. 

\subsection{Vectorisation of the chain map constraints}
\label{ssec:vectorisation}

We begin by rewriting the chain-map equations as a homogeneous linear system over $\mathbb{F}_2$. To do so, let $A$ and $B$ denote two CSS codes as presented in Sec.~\ref{ssec:codes_bg}, and fix bases for the spaces $C_i^A$ and $C_i^B$. We represent the unknown maps $\gamma_2$, $\gamma_1$, and $\gamma_0$ of the chain map $\boldsymbol{\gamma} : B_\bullet \xrightarrow{} A_\bullet$ by their matrix entries in these bases. We use the standard column-stacking vectorisation convention, whereby for a matrix $X \in \mathbb{F}_2^{m \times n}$, we denote $\operatorname{vec}(X)\in \mathbb{F}_2^{mn}$ the vector obtained by stacking the columns of $X$ from left to right. With this convention, we apply the identity 
\begin{equation}
\operatorname{vec}(S X Q)=(Q^{\mathsf T}\otimes S)\operatorname{vec}(X),
\label{eq:vec-identity}
\end{equation}
valid for matrices $S$, $X$, and $Q$ of compatible dimensions, to Eqs.~\eqref{eq:chain-map-conditions}. This yields the identities
\begin{align}
(I_{r_Z^B}\otimes (H_Z^A)^{\!\top})\,\operatorname{vec}(\gamma_2)
+
(H_Z^B\otimes I_{n_A})\,\operatorname{vec}(\gamma_1)
&=0,
\\
(I_{n_B}\otimes H_X^A)\,\operatorname{vec}(\gamma_1)
+
((H_X^B)^{\!\top}\otimes I_{r_X^A})\,\operatorname{vec}(\gamma_0)
&=0.
\end{align}
Stacking the equations above results in a homogeneous linear system $Mv = 0$ over $\mathbb{F}_2$, where 
\begin{equation}
M :=
\begin{pmatrix}
I_{r_Z^B}\otimes (H_Z^A)^{\!\top} & H_Z^B\otimes I_{n_A} & 0 \\
0 & I_{n_B}\otimes H_X^A & (H_X^B)^{\!\top}\otimes I_{r_X^A}
\end{pmatrix},
\qquad
v :=
\begin{pmatrix}
\operatorname{vec}(\gamma_2)\\
\operatorname{vec}(\gamma_1)\\
\operatorname{vec}(\gamma_0)
\end{pmatrix}.
\label{eq:M-block-form}
\end{equation}
Let us define the vectorisation map 
\begin{equation}
\mathcal{V}:
(\gamma_2,\gamma_1,\gamma_0)
\longmapsto
\begin{pmatrix}
\operatorname{vec}(\gamma_2)\\
\operatorname{vec}(\gamma_1)\\
\operatorname{vec}(\gamma_0)
\end{pmatrix}.
\label{eq:vec-map}
\end{equation}
\begin{theorem}
    Let $M$ be the matrix defined in~\eqref{eq:M-block-form}. Then, the map 
    \begin{equation}
        \mathcal{V}: \operatorname{Hom}_{\mathrm{Ch}}(B,A)\xrightarrow{\;\cong\;} \ker M 
    \end{equation}
    is a linear isomorphism. 
    \label{thm:M-kernel-isomorphism}
\end{theorem}
\begin{proof}
    The map $\mathcal{V}$ is linear and injective since given a $\boldsymbol{\gamma} =  (\gamma_2,\gamma_1,\gamma_0)$, if $\mathcal{V}(\boldsymbol{\gamma}) = 0$, then all three matrices have zero vectorisation and are therefore zero. If $\boldsymbol{\gamma} \in \operatorname{Hom}_{\mathrm{Ch}}(B,A)$, then the chain map equations hold, and by construction these are equivalent to the vector equation $M\mathcal{V}(\boldsymbol{\gamma})=0$, i.e. $\mathcal{V}(\boldsymbol{\gamma})\in \ker M$. The converse is straightforward by reshaping any vector $v \in \ker M$. Therefore $\mathcal{V}$ is bijective onto $\ker M$, and hence is a linear isomorphism.   
\end{proof}

While Theorem~\ref{thm:M-kernel-isomorphism} gives a complete description of $\operatorname{Hom}_{\mathrm{Ch}}(B,A)$, in many applications the middle component $\gamma_1$ is of direct interest. We therefore consider the projected space $\operatorname{Hom}_1(B,A)$ defined in Eq.~\eqref{eq:hom_1_space}. By Proposition~\ref{prop:hom1}, a matrix $\gamma_1 \in \mathbb{F}_2^{n_A\times n_B}$ lies in $\operatorname{Hom}_1(B,A)$ if and only if it satisfies the two subspace-inclusion conditions in Eqs.~\eqref{eq:hom1-cond1} and~\eqref{eq:hom1-cond2}. These can be expressed as a homogeneous linear system. Indeed, let $N_B$ be a matrix whose columns form a basis of $\ker\partial_1^B$, and let $U_A$ be a matrix whose columns form a basis of $\ker ((\partial_2^A{})^{\mathsf T})$. Then, these conditions are equivalent to 
\begin{align}
\partial_1^A\gamma_1N_B &= 0,\\
U_A^{\!\top}\gamma_1\partial_2^B &= 0.
\label{eq:hom1-conditions-vectorised}
\end{align}
By the vectorisation identity~\eqref{eq:vec-identity}, these give rise to the linear system $P\operatorname{vec}(\gamma_1)=0$ where 
\begin{equation}
    P = \begin{pmatrix}
        N_B^{\mathsf T} \otimes \partial_1^A \\
        (\partial_2^B)^{\mathsf T} \otimes U_A^{\mathsf T} 
    \end{pmatrix}. 
    \label{eq:hom1-linear-system}
\end{equation}
We obtain the immediate corollary. 
\begin{corollary}
    Under fixed bases and vectorisation convention, the restricted vectorisation map 
    \begin{equation}
        \operatorname{vec}\big |_{\operatorname{Hom}_1(B,A)}: \operatorname{Hom}_1(B,A) \xrightarrow{\;\cong\;} \ker P         
    \end{equation}
    is a linear isomorphism. 
\end{corollary}
If $\{\delta_i\}$ is a basis of $\ker P$ and $\Delta_i = \operatorname{vec}^{-1}(\delta_i)$, then $\{\Delta_i\}_{i=1}^q$ is a basis for $\HomI{B}{A}$. Hence, every $\gamma_1 \in \HomI{B}{A}$ admits a unique expansion 
\begin{equation}
    \gamma_1 = \sum_{i=1}^q x_i\Delta_i, \qquad x_i \in \mathbb{F}_2.
    \label{eq:hom1-basis-expansion}
\end{equation}
This provides a linear parametrisation of all chain-map feasible degree-1 maps between CSS codes $A$ and $B$, independent of any prescribed logical action or physical optimisation. This construction is summarised in Algorithm~\ref{alg:hom1-family}.    

\subsection{Affine restrictions of the degree-$1$ hom-space}
\label{ssec:affine-restrictions}

The basis expansion in Eq.~\eqref{eq:hom1-basis-expansion} provides a coordinate description of the linear space $\HomI{B}{A}$ of degree-1 maps that extend to full chain maps. We now consider additional constraints on its elements, which provide a framework for restricting $\HomI{B}{A}$ to linear or affine subspaces within this parametrisation. 

Let us fix a basis $\{\Delta_i\}_{i=1}^q$ of $\HomI{B}{A}$, such that every $\gamma_1$ in this space can be written as in Eq.~\eqref{eq:hom1-basis-expansion}. Let us then define $\Phi:\HomI{B}{A} \to \mathbb{F}_2^t$ as a linear map encoding a family of constraints, and $b\in\mathbb{F}_2^t$ a target vector. We define the corresponding feasible set by 
\begin{equation}
    \mathcal{A}_b := \{\gamma_1\in \HomI{B}{A} \mid \Phi(\gamma_1)=b\}.
    \label{eq:Ab-def}
\end{equation}
Since $\Phi$ is linear and $\HomI{B}{A}$ a linear space, $\mathcal{A}_b$ is either empty or an affine subspace of the degree-$1$ hom-space. If $\mathcal{A}_b\neq\varnothing$, and $\gamma_1^{(0)} \in \mathcal{A}_b$, it follows that
\begin{equation}
    \mathcal{A}_b = \gamma_1^{(0)} + \ker\Phi.
    \label{eq:affine-subspace}
\end{equation}
In particular, $\mathcal{A}_0=\ker \Phi$ is a linear subspace. In the coordinates of Eq.~\eqref{eq:hom1-basis-expansion}, the constraint $\Phi(\gamma_1)=b$ becomes a linear system
\begin{equation}
    \Lambda x=b,
    \qquad
    \Lambda \in \mathbb{F}_2^{t\times q}.
    \label{eq:Lx=b}
\end{equation}
where the $i^{\mathrm{th}}$ column of $\Lambda$ is $\Phi(\Delta_i)\in\mathbb{F}_2^t$. Hence, imposing linear constraints on admissible maps $\gamma_1$ reduces to solving a linear system over $\mathbb{F}_2$ in the coefficient vector $\boldsymbol{x}=\left(x_1,\dots,x_q\right)^{\mathsf T} \in \mathbb{F}_2^q$. If Eq.~\eqref{eq:Lx=b} admits a solution, let $x^{(0)}$ be a particular solution and let $\{\eta_j\}_{j=1}^m$ be a basis of $\ker \Lambda$. Writing 
\begin{equation}
    \widetilde{\Delta}_j := \sum_{i=1}^q (\eta_j)_i\,\Delta_i,
\end{equation}
every solution is uniquely of the form 
\begin{equation}
    \gamma_1 = \gamma_1^{(0)} + \sum_{j=1}^m y_j \widetilde{\Delta}_j, \qquad y_j\in\mathbb{F}_2, 
    \label{eq:affine-family}
\end{equation}
where 
\begin{equation}
    \gamma_1^{(0)} := \sum_{i=1}^q x_i^{(0)} \Delta_i. 
\end{equation}
In the following, we use this affine parametrisation as the starting point for logical targeting and optimisation. 

\section{Induced logical action}
\label{sec:logical-action}

The formalism developed in the previous section describes the space of admissible degree-$1$ components of a chain map between the chain complexes of two CSS codes. Let us now interpret this map as a physical interaction pattern, and describe its induced action on the logical operators.  

\subsection{Physical coupling}
\label{ssec:physical-coupling}

Let $A$ and $B$ be CSS codes with $n_A$ and $n_B$ physical qubits, represented as chain complexes as in Sec.~\ref{ssec:chain-maps-css}, with $\gamma_1: C_1^B \to C_1^A$. Let us fix bases for $C_1^A \cong \mathbb{F}_2^{n_A}$ and $C_1^B \cong \mathbb{F}_2^{n_B}$, such that the map $\gamma_1$ is represented by a binary matrix, denoted $\gamma_1\in\mathbb{F}_2^{n_A\times n_B}$.  We now interpret this matrix as specifying a pattern of CNOT gates from code $A$ to code $B$, as considered in Refs.~\cite{huang2023homomorphic, xu2025fast, li2025}. For each pair $(i,j)$ where $(\gamma_1)_{ij} = 1$, we include a CNOT gate with control on qubit $i$ of code $A$ and target on qubit $j$ of code $B$. Since all such CNOT gates act with controls supported on block $A$ and targets on block $B$, any pair of them commute. Together they define the physical unitary operator
\begin{equation}
    U_{\gamma_1} := \prod_{\substack{1\le i\le n_A\\1\le j\le n_B\\(\gamma_1)_{ij}=1}}
    \mathrm{CNOT}_{a_i\to b_j}.
    \label{eq:physical-cnot}
\end{equation}
where $a_i$ and $b_j$ indicate qubits of codes $A$ and $B$ respectively. For $x\in\mathbb{F}_2^n$ and $z\in\mathbb{F}_2^n$, let us write 
\begin{equation}
    X(x) := \bigotimes_{i=1}^n X_i^{x_i},
    \qquad 
    Z(z) := \bigotimes_{i=1}^n Z_i^{z_i}.
\end{equation}
Then, conjugation by $U_{\gamma_1}$ acts on physical Pauli operators according to 
\begin{align}
    U_{\gamma_1}X_A(x_a)U_{\gamma_1}^{\dagger} &= X_A(x_a)X_B(\gamma_1^{\mathsf T}x_a), 
    \label{eq:phys-X}\\
    U_{\gamma_1}Z_B(z_b)U_{\gamma_1}^{\dagger} &= Z_A(\gamma_1z_b)Z_B(z_b) 
    \label{eq:phys-Z}
\end{align}
for all $x_a \in \mathbb F_2^{n_A}$ and $z_b \in \mathbb F_2^{n_B}$, where $X_A$ and $Z_B$ respectively denote Pauli operators with support on codes $A$ and $B$ and $\gamma_1^{\mathsf T}: C_1^A\to C_1^B$. Hence, $X$-type information propagates from code $A$ to code $B$ through $\gamma_1^{\mathsf T}$, while $Z$-type information propagates from code $B$ to code $A$ through $\gamma_1$. Moreover, the chain-map conditions imply that $Z$-stabilizers of code $B$ are mapped to $Z$-stabilizers of code $A$, and that $X$-stabilizers of code $A$ are mapped to $X$-stabilizers of code $B$. Equivalently, conjugation by $U_{\gamma_1}$ preserves the product codespace $\mathcal{C}_A\otimes \mathcal{C}_B$. 

\subsection{Induced logical \texorpdfstring{$Z$}{Z}- and \texorpdfstring{$X$}{X}-actions}

We now describe the logical action induced by the physical operator $U_{\gamma_1}$ in Eq.~\eqref{eq:physical-cnot}. A degree-$1$ chain map component $\gamma_1: C_1^B\to C_1^A$ induces a map on homology $\gamma_\star$ by Lemma~\ref{lemma:induced_op}. The induced map $\gamma_\star$ may be represented in chosen logical
bases by a matrix $\Gamma_Z(\gamma_1)$ defined as follows. Since $H_1$ identifies the $Z$-type logical operators of a CSS code, this map determines the propagated $Z$-type logical action of the physical coupling described in Sec.~\ref{ssec:physical-coupling}. Indeed, let $\boldsymbol L_Z^A\in\mathbb F_2^{k_A\times n_A}$ and
$\boldsymbol L_Z^B\in\mathbb F_2^{k_B\times n_B}$ be matrices whose rows are chosen representatives of bases of $H_1(A)$ and $H_1(B)$, respectively. The matrix $\gamma_1(\boldsymbol L_Z^B)^{\mathsf T}\in \mathbb F_2^{n_A\times k_B}$ has columns in $\ker \partial_1^A$, and hence defines $k_B$ homology classes in $H_1(A)$. Hence, there exists a unique matrix $\Gamma_Z(\gamma_1)\in \mathbb F_2^{k_A\times k_B}$ such that 
\begin{equation}
    \gamma_1(\boldsymbol L_Z^B)^{\mathsf T} = \partial_2^A\beta_Z + (\boldsymbol L_Z^A)^{\mathsf T}\Gamma_Z(\gamma_1),
    \label{eq:GammaZ-def}
\end{equation}
for some coefficient matrix $\beta_Z \in \mathbb F_2^{r_Z^A\times k_B}$. Dually, the propagated logical $X$-type action is obtained from the transpose map $\gamma_1^{\mathsf T}$ which acts on logical $X$ operators. With $\boldsymbol L_X^A\in\mathbb F_2^{k_A\times n_A}$ and $\boldsymbol L_X^B\in\mathbb F_2^{k_B\times n_B}$ matrices whose rows are representatives of $X$-type logical operators of codes $A$ and $B$, there exists a unique matrix $\Gamma_X(\gamma_1) \in \mathbb F_2^{k_B\times k_A}$ such that 
\begin{equation}
    \gamma_1^{\mathsf T}(\boldsymbol L_X^A)^{\mathsf T} = (\partial_1^B)^{\mathsf T} \beta_X + (\boldsymbol L_X^B)^{\mathsf T} \Gamma_X(\gamma_1),
    \label{eq:GammaX-def}
\end{equation}
for some coefficient matrix $\beta_X \in \mathbb F_2^{r_X^B \times k_A}$. Thus, for fixed logical bases, the coupling map $\gamma_1$ determines two linear maps 
\begin{equation}
    \Gamma_Z: \mathbb F_2^{k_B} \to \mathbb F_2^{k_A}, \qquad \Gamma_X: \mathbb F_2^{k_A} \to \mathbb F_2^{k_B},
\end{equation}
respectively describing the induced logical $Z$- and $X$-propagation. 
\begin{proposition}
Let $\gamma_1\in \HomI{B}{A}$, and let
\[
\Gamma_Z(\gamma_1)\in \mathbb F_2^{k_A\times k_B},
\qquad
\Gamma_X(\gamma_1)\in \mathbb F_2^{k_B\times k_A}
\]
be the induced logical $Z$- and $X$-action matrices in the chosen logical bases. Define
\[
\Omega_A:=\boldsymbol L_X^A(\boldsymbol L_Z^A)^{\mathsf T}\in\mathbb F_2^{k_A\times k_A},
\qquad
\Omega_B:=\boldsymbol L_X^B(\boldsymbol L_Z^B)^{\mathsf T}\in\mathbb F_2^{k_B\times k_B}.
\]
Then
\begin{equation}
\Omega_A\Gamma_Z(\gamma_1)
=
\Gamma_X(\gamma_1)^{\mathsf T}\Omega_B.
\label{eq:Gamma-compat}
\end{equation}
In particular, if the logical bases are symplectically normalised so that $\Omega_A=I_{k_A}$ and $\Omega_B=I_{k_B}$, then Eq.~\eqref{eq:Gamma-compat} reduces to 
\begin{equation}
    \Gamma_Z(\gamma_1) = \Gamma_X(\gamma_1)^{\mathsf T}. 
    \label{eq:Gamma-compat-symplectic}
\end{equation}
\label{prop:Gamma-compat}
\end{proposition}
\begin{proof}
For every $i\in\{1,\dots,k_A\}$ and $j\in\{1,\dots,k_B\}$, conjugation by the Clifford operator $U_{\gamma_1}$ preserves the Pauli commutation pairing between the $i^{\mathrm{th}}$ logical $X$ operator of $A$ and the $j^{\mathrm{th}}$
logical $Z$ operator of $B$. Since these operators initially act on different code blocks, their pairing is trivial, and hence so is the pairing of their propagated images. Computing this propagated pairing in the chosen logical bases gives
\[
e_i^{\mathsf T}\Omega_A\Gamma_Z(\gamma_1)e_j
=
e_i^{\mathsf T}\Gamma_X(\gamma_1)^{\mathsf T}\Omega_B e_j.
\]
Since this holds for all $i,j$, Eq.~\eqref{eq:Gamma-compat} follows. The
symplectic-normalised case is immediate. 
\end{proof}
Note that since a change of logical basis of codes $A$ and $B$ acts on $\Gamma_Z(\gamma_1)$ and $\Gamma_X(\gamma_1)$ by invertible left-right multiplication, properties of the induced logical action such as rank, injectivity, surjectivity, and bijectivity are basis-independent. This motivates the following classification of admissible couplings $\gamma_1\in\HomI{B}{A}$.
\begin{definition}
    The logical rank of a coupling $\gamma_1\in\HomI{B}{A}$ is
    \begin{equation}
        r(\gamma_1):=\dim(\operatorname{im}\gamma_\star)
        =\operatorname{rank}\Gamma_Z(\gamma_1),
    \end{equation}
    where $\gamma_\star:H_1(B)\to H_1(A)$ denotes the induced map on homology.
\end{definition}
The rank $r(\gamma_1)$ counts the number of independent logical $Z$ degrees of freedom of code $B$ that propagate non-trivially to logical $Z$ degrees of freedom of code $A$. In particular, $r(\gamma_1) = 0$ corresponds to a logically trivial coupling, whereas $r(\gamma_1) = \min(k_A, k_B)$ is the maximal possible logical rank. 

\begin{definition}
Let $\gamma_1\in\HomI{B}{A}$ and $\gamma_\star: H_1(B) \to H_1(A)$ be the induced map on homology. We say that $\gamma_1$ is logically
\begin{enumerate}
    \item injective if $\gamma_\star$ is injective, i.e. $\ker\Gamma_Z(\gamma_1)=\{0\}$,
    \item surjective if $\gamma_\star$ is surjective, i.e. $\operatorname{im}\Gamma_Z(\gamma_1)=\mathbb F_2^{k_A}$,
    \item bijective if $\gamma_\star$ is bijective, i.e. $k_A=k_B$ and $\Gamma_Z(\gamma_1)$ is invertible.
\end{enumerate}
\end{definition}
Injective couplings realise logical embeddings of the $Z$-sector of code $B$ into that of code $A$, while surjective couplings realise logical projections onto the $Z$-sector of code $A$. By Proposition~\ref{prop:Gamma-compat}, the induced $X$-sector action is then constrained by commutation preservation. In symplectically normalised bases, Proposition~\ref{prop:Gamma-compat} identifies the compatible logical $X$ action via Eq.~\eqref{eq:Gamma-compat-symplectic}. 

Up to independent changes of logical basis on code $A$ and $B$, any matrix $\Gamma_Z(\gamma_1)\in\mathbb F_2^{k_A\times k_B}$ of rank $r$ can be brought to the canonical form 
\begin{equation}
S\,\Gamma_Z(\gamma_1)\,T =
    \begin{pmatrix}
    I_r & 0\\
    0 & 0
    \end{pmatrix},
    \qquad
    S\in GL(k_A,\mathbb F_2),\quad T\in GL(k_B,\mathbb F_2).
\label{eq:GammaZ-normal-form}
\end{equation}
Hence, up to independent changes of logical basis on $A$ and $B$, the induced logical $Z$-action is equivalent to $r$ independent logical CNOT couplings, together with spectator logical operators. This motivates the following definition.
\begin{definition}
A coupling $\gamma_1\in\HomI{B}{A}$ is said to realise a rank-$r$ homomorphic \emph{CNOT} if 
\begin{equation}
    \operatorname{rank}\Gamma_Z(\gamma_1)=r.
    \label{eq:rank-r-hCNOT}
\end{equation}
It is said to realise a full-rank homomorphic \emph{CNOT} if \[
\operatorname{rank}\Gamma_Z(\gamma_1)=\min(k_A,k_B).\]
\end{definition}
A full-rank homomorphic CNOT need not be bijective unless $k_A=k_B$. When $k_A=k_B$ and $\Gamma_Z(\gamma_1)$ is invertible, the chain map induces a one-to-one pairing between the logical-$Z$ sectors of the two codes and is equivalent, up to logical basis changes, to $k_A$ independent logical CNOTs.

\section{Targeting logical action as affine constraints}
\label{sec:logical-targeting}

\subsection{General affine targeting}
\label{ssec:general-scheme}

We now specialise the affine-restriction framework of Sec.~\ref{ssec:affine-restrictions} to the case in which constraints are imposed on the induced logical action of a degree-$1$ map. After fixing logical representatives for codes $A$ and $B$, each admissible coupling $\gamma_1 \in \HomI{B}{A}$ determines induced logical matrices $\Gamma_Z\in\mathbb F_2^{k_A\times k_B}$ and $\Gamma_X\in\mathbb F_2^{k_B\times k_A}$, as defined in Eqs.~\eqref{eq:GammaZ-def} and~\eqref{eq:GammaX-def}. By construction, both depend linearly on $\gamma_1$. Let us combine them into the logical action map 
\begin{equation}
\Gamma:
\HomI{B}{A}\to
\mathbb F_2^{k_A\times k_B}\oplus \mathbb F_2^{k_B\times k_A},
\qquad
\Gamma(\gamma_1):=\bigl(\Gamma_Z(\gamma_1),\Gamma_X(\gamma_1)\bigr).
\label{eq:logical-action-map}
\end{equation}
Let 
\begin{equation}
    \mathcal{T}:\mathbb F_2^{k_A\times k_B}\oplus \mathbb F_2^{k_B\times k_A}\to \mathbb F_2^t
\end{equation}
be a linear map, and let $b \in \mathbb F_2^t$. We define the corresponding logically targeted family by 
\begin{equation}
\HomI{B}{A}\big|_{(\mathcal T,b)}
:= \left\{
\gamma_1\in \HomI{B}{A} \;\middle|\;\mathcal T\!\bigl(\Gamma_Z(\gamma_1),\Gamma_X(\gamma_1)\bigr)=b\right\}.
\label{eq:logical-target-family}
\end{equation}
Since both $\Gamma$ and $\mathcal T$ are linear, their composition $\mathcal T \circ \Gamma$ is also linear. Hence, $\HomI{B}{A}\big|_{(\mathcal T,b)}$ is either empty or an affine subspace of $\HomI{B}{A}$. Thus, targeting a logical action is a particular instance of the affine-restriction formalism introduced in Sec.~\ref{ssec:affine-restrictions}. A particularly interesting case is obtained by requiring that
\begin{equation}
    \Gamma_Z(\gamma_1)=\Gamma_Z^{\mathrm{tgt}},
    \qquad
    \Gamma_X(\gamma_1)=\Gamma_X^{\mathrm{tgt}},
\end{equation}
for desired targets $\Gamma_Z^{\mathrm{tgt}}\in \mathbb F_2^{k_A\times k_B}$, and $\Gamma_X^{\mathrm{tgt}}\in \mathbb F_2^{k_B\times k_A}$. This is a logically targeted family of the form~\eqref{eq:logical-target-family}. With normalised symplectic logical bases, Proposition~\ref{prop:Gamma-compat} implies that any valid target must satisfy
\begin{equation}
    \Gamma_X^{\mathrm{tgt}}=(\Gamma_Z^{\mathrm{tgt}})^{\mathsf T}.
    \label{eq:clifford-tgt}
\end{equation}
Hence, in that setting, it is sufficient to prescribe the induced logical $Z$-action alone. In our implementation, we take the target map $\mathcal T$ to be the vectorisation of $\Gamma_Z$, with
$b=\operatorname{vec}(\Gamma_Z^{\mathrm{tgt}})$, and the compatible
$X$-action is fixed by Eq.~\eqref{eq:clifford-tgt}. Algorithm~\ref{alg:restrict-by-gammaxz} summarises this restriction of an
affine family to a prescribed logical action.

A useful special case is obtained by specifying which logical degrees of freedom are allowed to participate in the induced action. Let
\[
S_A\subseteq \{1,\dots,k_A\},
\qquad
S_B\subseteq \{1,\dots,k_B\}
\]
be subsets of the logical indices of codes $A$ and $B$, respectively. A \emph{subset mapping} from $S_B$ to $S_A$ is enforced by the linear support constraints
\begin{equation}
    (\Gamma_Z(\gamma_1))_{ij}=0,
    \qquad
    \forall (i,j)\notin S_A\times S_B.
    \label{eq:subset-support}
\end{equation}
Additional linear conditions may then be imposed on the entries of $\Gamma_Z(\gamma_1)$ within the allowed block $S_A\times S_B$. Hence, subset mapping is a genuine affine restriction of the logical action. 

\subsection{Structured nonlinear targets}
\label{ssec:nonlinear-targeting}

Not every logically meaningful requirement is affine. In practice, such conditions may be handled either by fixing logical bases and replacing the nonlinear requirement by an equivalent prescribed target matrix whenever a normal form is available, or by imposing the nonlinear condition after constructing an affine family of admissible candidates. For example, 
\begin{equation}
\operatorname{rank}\Gamma_Z(\gamma_1)=r,
\qquad
\operatorname{rank}\hspace{1mm}\!\bigl(\Gamma_Z(\gamma_1)|_{S_A\times S_B}\bigr)=1
\end{equation}
are nonlinear in $\gamma_1$, and therefore do not, in general, define affine subspaces of $\HomI{B}{A}$. For rank-$r$ targets, the normal form in Eq.~\eqref{eq:GammaZ-normal-form} shows that any rank-$r$ logical $Z$-action is equivalent, under independent logical basis changes, to the canonical matrix
\begin{equation}
    \begin{pmatrix}
        I_r & 0\\
        0 & 0
    \end{pmatrix}.
\end{equation}
Accordingly, under freedom of choice of logical bases, one may realise rank-$r$ homomorphic CNOTs by fixing such a representative and solving the corresponding affine targeting problem. With fixed logical bases, this selects a single affine slice of the generally nonlinear rank-$r$ condition.

A second example is obtained by requiring that the induced logical action on a selected logical block have rank $1$. This motivates the following notion of a logical FANOUT.

\begin{definition}
Let $S_A\subseteq \{1,\dots,k_A\}$ and
$S_B\subseteq \{1,\dots,k_B\}$. A map
$\gamma_1 \in \HomI{B}{A}$ is said to realise a logical
\emph{\textbf{FANOUT}} if
\begin{equation}
\operatorname{supp}\Gamma_Z(\gamma_1)\subseteq S_A\times S_B,
\qquad
\operatorname{rank}(
\Gamma_Z(\gamma_1)|_{S_A\times S_B}
)=1.
\label{eq:FANOUT-def}
\end{equation}
\end{definition}
More generally, this logical targeting formalism extends to collections of CSS code blocks, where both $\gamma_1$ and the induced logical matrices acquire natural block decompositions, and affine targeting may be imposed on individual block pairs or on collections of block interactions. This is described in Appendix~\ref{assec:multi-block-targeting}.

\section{Optimisation}
\label{sec:optimisation}

\subsection{CP-SAT formulation}
\label{sec:cp-sat}

The formalism in Secs.~\ref{ssec:affine-restrictions}~and~\ref{ssec:general-scheme} reduces the construction of a logically targeted inter-code coupling to the selection of an element of an affine family of admissible degree-$1$ maps. Concretely, let
\begin{equation}
    \mathcal A = \gamma_1^{(0)} + \operatorname{span}\{\widetilde{\Delta}_1, \dots, \widetilde{\Delta}_m\} \subseteq \HomI{B}{A},
\end{equation}
denote such a non-empty affine family. Then, every element $\gamma_1(y) \in \mathcal A$ is of the form in Eq.~\eqref{eq:affine-family}, where $y=(y_1, \dots, y_m)^{\mathsf T} \in \mathbb F_2^m$. Thus, the algebraic validity and logical targeting constraints are baked into the affine search space, and it remains to select, among its elements, a representative that is practically favourable. Indeed, given a cost function $f : \mathbb F_2^{n_A\times n_B}\to \mathbb R$, which quantifies the quality of a coupling map $\gamma_1$, we cast the synthesis of such a map as a combinatorial optimisation problem, i.e. $\min_{y\in \mathbb F_2^m} f(\gamma_1(y))$, over the admissible affine solution set. 

In practice, a $\gamma_1$ is interpreted as a physical CNOT circuit defined in Eq.~\eqref{eq:physical-cnot}. We therefore seek a solution whose associated circuit is \emph{shallow}, and preferably \emph{sparse}. While these conditions do not guarantee code distance-preservation, they may serve as a useful proxy by virtue of, respectively, minimising the time overhead, and number of entangling gates in the operation, which helps in limiting the spread of errors in the circuit, and thereby reduce the cost of fault-tolerant flagging. In our implementation, these costs are derived from the bipartite interaction graph encoded by $\gamma_1$, and optionally support constraints which may reflect geometry or connectivity by using a \emph{mask} $\mathcal B \in \mathbb F_2^{n_A\times n_B}$ specifying allowed couplings, with $\mathcal B_{ij}$ excluding edge $(a_i, b_j)$. 

Let $\gamma_1 \in \mathbb F_2^{n_A\times n_B}$ be a feasible map, specifying a set of physical CNOT gates from code blocks $A$ to $B$, with an edge $(i,j)$ whenever $(\gamma_1)_{ij} = 1$. We therefore associate to this map the bipartite graph 
\begin{equation}
G_{\gamma_1} = (V_A \sqcup V_B, E_{\gamma_1}), \qquad E_{\gamma_1} = \{(a_i, b_j) \in V_A \times V_B \text{ | } (\gamma_1)_{ij} = 1 \},
\end{equation}
where $V_A = \{a_1, \dots, a_{n_A}\}$ and $V_B = \{b_1, \dots, b_{n_B}\}$ index the qubits of codes $A$ and $B$, respectively. The total number of physical CNOT gates is therefore $|E_{\gamma_1}|$, which we refer to as the \emph{weight} $w(\gamma_1)$ of the coupling. Let us define the maximum row and column degrees 
\begin{equation}
    \qquad \Delta_A(\gamma_1) = \max_{i \in V_A} \sum_j (\gamma_1)_{ij}, \qquad \Delta_B(\gamma_1) = \max_{j \in V_B} \sum_i (\gamma_1)_{ij}, 
\end{equation}
such that the overall maximum degree of $\gamma_1$ is given by
\begin{equation}
    \Delta(\gamma_1):= \max \left\{ \Delta_A(\gamma_1), \Delta_B(\gamma_1) \right\}.
\end{equation}
Since all gates in $U_{\gamma_1}$ are CNOTs from code $A$ to $B$, two such gates can be executed in parallel if and only if they do not share a physical qubit. Equivalently, a valid schedule implementing $U_{\gamma_1}$ is a partition of the edge set $E_{\gamma_1}$ into disjoint matchings of the graph $G_{\gamma_1}$, where each matching corresponds to one time-step. The minimum number of time steps required is given by the edge-chromatic number of the bipartite graph $G_{\gamma_1}$. Proposition~\ref{prop:gamma1-depth} follows by K\"onig's line colouring theorem.
\begin{proposition}
    Let $\gamma_1 \in \mathbb F_2^{n_A\times n_B}$ be a feasible coupling, and let $G_{\gamma_1}$ be its associated bipartite interaction graph. Then, the minimum circuit depth required to realise the unitary $U_{\gamma_1}$ is $\Delta(\gamma_1)$. 
    \label{prop:gamma1-depth}
\end{proposition}
Thus, depth minimisation is equivalent to minimising the largest number of incident gates on any physical qubit, and the optimisation problem can be instantiated with objectives such as $f(\gamma_1) = \Delta(\gamma_1)$, or lexicographically by minimising $\Delta(\gamma_1)$ and then weight $w(\gamma_1)$ to bias towards sparser implementations among all minimum-depth schedules. This optimisation problem can be encoded in a form suitable for constraint programming, by exploiting the affine parametrisation of $\mathcal A$. Given a $\gamma_1$ of the form in Eq.~\eqref{eq:affine-family}, the coefficients $y_t$ are promoted to Boolean decision variables in the CP-SAT model, and the search for an optimal coupling corresponds to a search over binary assignments to the affine coordinates of $\mathcal A$. For each matrix entry $(\gamma_1)_{ij}$, we introduce an auxiliary Boolean variable $g_{ij}$, intended to represent the value of the corresponding entry in the coupling matrix $\gamma_1(y)$. By Eq.~\eqref{eq:affine-family}, this value is the parity 
\begin{equation}
    g_{ij} = (\gamma_1^{(0)})_{ij} \oplus \bigoplus \limits_{\substack{t \\ (\widetilde{\Delta}_t)_{ij}=1}} y_t
\end{equation}
where $\oplus$ is over $\mathbb F_2$, such that each $g_{ij}$ is constrained by one XOR relation involving the affine coefficient variables. The physical cost variables defined above are then linearly expressed in terms of $g_{ij}$. In particular, the row and column participation numbers, and weight are 
\begin{equation}
    r_i = \sum_{j=1}^{n_B} g_{ij}, \qquad c_j = \sum_{i=1}^{n_A} g_{ij}, \qquad W = \sum_{i,j} g_{ij}.
\end{equation}
The circuit-depth is represented by an integer variable $D$, with the constraints
\begin{equation}
    r_i \leq D, \quad c_j \leq D, \qquad \forall i\in V_A, \quad  j\in V_B.
\end{equation}
By Proposition~\ref{prop:gamma1-depth}, any admissible assignment satisfies $D \geq \Delta(\gamma_1)$, and minimising $D$ is equivalent to minimising the depth of the associated CNOT circuit. Optional support constraints can be incorporated by fixing selected variables $g_{ij}$ to zero using a mask. In the present implementation, we primarily minimise $D$, optionally imposing
bounds on weight, row degree, or column degree. We also tested weighted
objectives such as $\lambda D+W$, with $\lambda>n_An_B$ sufficiently large to ensure that depth-reduction is preferred over possible increase in weight. In practice, aggressively minimising weight can dramatically reduce the fault-distance of the circuit. Hence, we also consider depth-only optimisation. 

Finally, the formulation above depends on the chosen basis of the linear part of the affine family. Although different bases define the same feasible set, they need not be computationally equivalent from the perspective of constraint programming. This motivates the sparse basis replacement heuristic described next.

\subsection{Sparse basis replacement for admissible degree-$1$ families}
\label{ssec:sparse-basis}

The algorithm~\ref{alg:hom1-family} described above yields an exact but otherwise arbitrary basis of $\HomI{B}{A}$. Therefore, the resulting generators can be dense. Sparse bases for null spaces are well-studied objects~\cite{coleman1986null, gilbert1987computing, coleman1987null}. In our setting, replacing dense generators by sparser and more structured ones can be advantageous heuristically, since it can expose additional combinatorial structure in the search space consisting of an already constructed affine family, and exact optimisation performance is known to depend strongly on how the search space is represented and explored~\cite{huang2021branch}. Indeed, consider a previously constructed affine family of admissible couplings as described in~\eqref{eq:affine-family}. A dense basis means that one $x_i$ flips many entries in $\gamma_1$ and therefore perturbs many row and column sums at once. Sparse basis replacement reduces the support of individual search directions, and therefore reduces the coupling between coordinate variables $x_i$ and the bottleneck row/column degree constraints. If additionally the basis has some structure, then each $x_i$ acts more locally, which provides the solver with a more decomposed and interpretable search space. As a basic heuristic, we adopt a replacement basis adapted to the parity-check structure of codes $A$ and $B$. Specifically, for an affine family $\mathcal A$ as presented in Eq.~\eqref{eq:affine-family}, let
\begin{equation}
    V:=\operatorname{span}\{\widetilde{\Delta}_1,\dots,\widetilde{\Delta}_m\}\subseteq \mathbb F_2^{n_A\times n_B},
\end{equation}
denote its linear part. Let $h_p \in\mathbb F_2^{n_B}$ and $z_s \in F_2^{n_A}$ 
denote the column-vector transposes of the $p^{\mathrm{th}}$ row of $H_X^B$ and the $s^{\mathrm{th}}$ row of $H_Z^A$, respectively. We consider candidates consisting of rank-$1$ matrices of the form 
\begin{align}
    e_i h_p^{\mathsf T},& \qquad e_i\in\mathbb F_2^{n_A},\quad h_p\in\mathbb F_2^{n_B},\\
    z_s e_j^{\mathsf T},& \qquad z_s\in\mathbb F_2^{n_A},\quad e_j\in\mathbb F_2^{n_B},
\end{align}
where $e_i, e_j$ are column vectors. The choice of rows of $H_X^B$ and $H_Z^A$ reflects the directional Pauli propagation of a CNOT coupling from $A$ to $B$, and corresponds to row-local or column-local moves dressed by parity-check patterns. For instance, a move $e_i h_p^{\mathsf T}$ means that row $i$ of $\gamma_1$, i.e.\ qubit $i$ in $A$, is coupled to the qubits of $B$ appearing in the $X$-check $h_p$. Similarly, a move $z_s e_j^{\mathsf T}$ couples qubit $j$ in $B$ to exactly those qubits in $A$ lying on the support of the $Z$-stabilizer $z_s$. Among these candidates, we retain only those that already lie in the linear space $V$, and then extract a basis of their span. If this span is all of $V$, the result is an equivalent reparametrisation of the original affine family, i.e. a change of the coordinate description of the search space. Otherwise, we retain the original parametrisation. This heuristic is summarised in Algorithm~\ref{alg:sparse-basis}. In practice, we also use these sparse moves as local coordinates for the solver, while retaining a small number of more global generators that capture coordinated deformations of $\gamma_1$ which are not well captured by purely row- or column-local moves. This balances the minimal-support search directions with added global freedom to navigate the affine family efficiently.

\section{Applications and benchmarks}
\label{sec:applications}

Having developed the theoretical framework of homomorphic CNOTs from chain maps and their discovery, we now present concrete applications for which we obtain favourable spacetime tradeoffs as compared to previous techniques.     

\begin{table}[ht!]
\centering
\begin{tabular}{l l l c c c}
\hline
Code $A$ & Code $B$ & $\Delta$ & $w$ & $(d_X^{\rm gad},d_Z^{\rm gad})$ & - \\
\hline
6.6.6 color-3 & surface-3 & 2 & 9 & $(3, 3)$ & \\
6.6.6 color-5 & surface-5 & 2 & 27 & $(5, 5)$ & \\
6.6.6 color-7 & surface-7 & 2 & 59 & $(7, 7)$ & \\
\hline
surface-4 & $[[36,8,4]]$ BB & 2 & 16 & $(4, 4)$ & \\
surface-5 & $[[36,8,4]]$ BB & 3 & 20 & $(4, 4)$ & \\
surface-5 & $[[30,8,4]]$~\cite{webster2026} & 3 & 34 & $(4, 4)$ & \\
surface-6 & $[[20,2,6]]$~\cite{berthusen2025} & 2 & 36 & $(6, 6)$ & \\
          & $[[72,12,6]]$~\cite{Bravyi_2024} & 4 & 86 & $(4, 6)$ & \\
surface-7 & $[[72,12,6]]$~\cite{berthusen2025} & 4 & 120 & $(4, 6)$ & \\
\hline
6.6.6 color-5 & $[[30,8,4]]$~\cite{webster2026} & 3 & 28 & $(4,4)$ & \\
 & $[[36,8,4]]$ BB & 3 & 32 & $(4, 4)$ & \\
 & $[[23, 1, 7]]$~\cite{jain2025transversal} & 3 & 44 & $(5,6)$ & \\
& $[[54,14,5]]$~\cite{galimova2026independent} & 4 & 47 & $(5,3)$ & \\
\hline
$[[27,3,3]]$~\cite{li2025} & $[[18, 2, 3]]$ & 1 & 18 & $(3, 3)$ & \\
                           & Steane & 2 & 11 & $(3, 3)$ & \\
                           & $\textrm{Steane}^{\oplus 3}$ & 2 & 33 & $(3, 3)$ & \\
                           & surface-3 & 2 & 11 & $(3, 3)$ & \\
                           & $[[10,2,3]]$ & 2 & 17 & $(3, 3)$ & \\
\hline
$[[17,1,5]]$~\cite{jain2025transversal} & $[[23,1,7]]$ & 3 & 51 & $(5, 5)$ & \\
 & $[[20,2,6]]$~\cite{berthusen2025} & 2 & 26 & $(5, 5)$ & \\
 & surface-$5$ & 2 & 25 & $(5, 5)$ & \\
\hline
$[[15,1,3]]$ & surface-3 & 2 & 9 & $(3, 3)$ & \\
 & Steane & 1 & 7 & $(3, 3)$ & \\
             & $[[15, 7, 3]]$ & 2 & 17 & $(3, 3)$ & \\
             & $[[24, 8, 3]]$~\cite{gu2026qgpu} & 2 & 21 & $(3, 3)$ & \\
             & $[[16, 6, 4]]$~\cite{tansuwannont2026} & 2 & 20 & $(3, 4)$ & \\
             & $[[23, 1, 7]]$~\cite{jain2025transversal} & 1 & 15 & $(3, 7)$ & \\
             & surface-7 & 2 & 21 & $(3, 7)$ & \\
\hline
$[[36,4,4]]$~\cite{dasu2025} & $[[16,4,4]]$~\cite{Goto_2024} & 3 & 36 & $(4,4)$ & \\
                             & $[[30,8,4]]$~\cite{webster2026} & 5 & 114 & $(4,4)$ & \\
\hline
$[[49,1,5]]$~\cite{Bravyi_2012} & surface-5 & 2 & 25 & $(5,5)$ & \\
                                & $[[20,2,6]]$~\cite{berthusen2025} & 2 & 30 & $(5,6)$ & \\
 & & & & & \\ 
\hline
\multicolumn{6}{l}{\textbf{Flagged homomorphic CNOT circuits}} \\
Code $A$ & Code $B$ & $\Delta$ & $w$ & $(d_X^{\rm gad},d_Z^{\rm gad})$ & $n_{\mathrm{flags}}$ \\
\hline
surface-$7$ & $[[42, 10, 6]]$ BB & 4 & 158 & $(6, 6)$ & 27 \\
\hline
surface-$7$ & $[[72, 12, 6]]$~\cite{Bravyi_2024} & 4 & 130 & $(6, 6)$ & 5 \\
\hline
color-$7$ & $[[62, 10, 6]]$~\cite{webster2025} & 5 & 164 & $(6, 6)$ & 18 \\
\hline
\end{tabular}
\caption{\small \textbf{Examples of homomorphic CNOT circuits derived from chain maps.} Here, the chain maps have full rank, i.e. $r=\min (k_A, k_B)$ and realise $r$ CNOTs between logical qubits in order $(1,2,...,r)$, $\Delta$ indicates the circuit depth of the corresponding physical unitary, and $w$ the number of physical CNOTs. Circuit-level distances were verified to be $(d_X^{\rm gad},d_Z^{\rm gad})$ with \texttt{Stim}~\cite{Gidney_2021}. The $\oplus l$ superscript indicates $l$ code blocks. Here, BB indicates a bivariate bicycle code~\cite{Bravyi_2024}. The number of gates $w$ for the flagged circuit includes the CNOT pairs required to couple the $n_{\mathrm{flag}}$ flag ancillae.}
\label{tab:circ_table}
\end{table}

\subsection{Heterogeneous logic}

The most immediate application of bespoke chain maps is to facilitate the interaction of dissimilar codes in a heterogeneous architecture. Such interactions are heavily used in QECC architectures in which different operations, i.e. Clifford and non-Clifford gates and gadgets, are performed on different codes~\cite{Sullivan_2024, tan2025, yoder2025, webster2026, cain2026}. Many of these architectures leverage codes that are closely related structurally, such as between 2D and 3D codes; however, restricting oneself to codes that are related in this way greatly limits the available pairings. Previously, the only way to implement arbitrary inter-code interactions was to construct a universal adapter~\cite{swaroop2026} and perform joint logical measurements across the two code blocks. For example, a $\ket{T}$ magic state would then be injected onto a given logical qubit $\ket{\psi}$ using the following sequence of Pauli measurements and Clifford corrections:

\begin{figure}[h]
    \centering
    \includegraphics[width=0.9\linewidth]{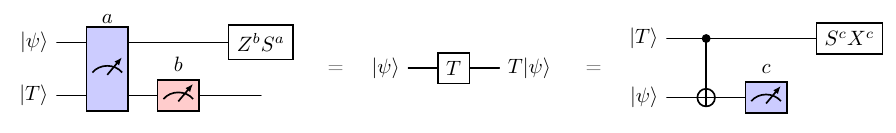}
    \caption{\small \textbf{Methods to inject magic states.} (Left) Logical circuit to inject $\ket{T}$ magic states using universal adapters~\cite{swaroop2026}. The blue measurement with measurement output $a$ is a joint ZZ measurement. The red measurement with measurement output b is an X measurement. (Right) circuit to inject $\ket{T}$ states using an inter-block logical CNOT gate.}
    \label{fig:placeholder}
\end{figure}

While flexible, universal adapters to measure a weight-$t$ joint logical Pauli measurement require a number of ancilla qubits scaling as $\tilde{O}(td)$, as well as $O(d)$ rounds of measurement, to satisfy fault-tolerance. Such overheads are substantial in practice. The alternative to joint $ZZ$ measurements for implementing state injection is to use CNOT gates. This is the approach taken in, for example, architectures based on code switching~\cite{anderson2014, tan2025}, or transversal surface code architectures~\cite{sunami2025}. While more efficient than adapters in terms of spacetime overheads, the available code pairings are limited. In contrast, using bespoke chain maps allows us to achieve the flexibility of adapters while maintaining overheads that are closer to transversal gates.

\begin{figure}[!t]
    \centering
    \begin{subfigure}[b]{0.45\textwidth}
        \centering
        \includegraphics[width=\textwidth]{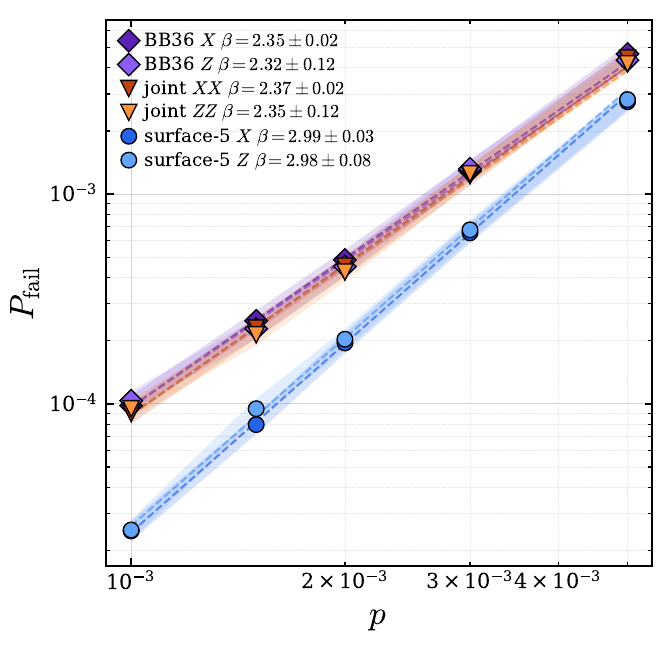}
        \caption{}
    \end{subfigure}
    \hfill
    \begin{subfigure}[b]{0.45\textwidth}
        \centering
        \includegraphics[width=\textwidth]{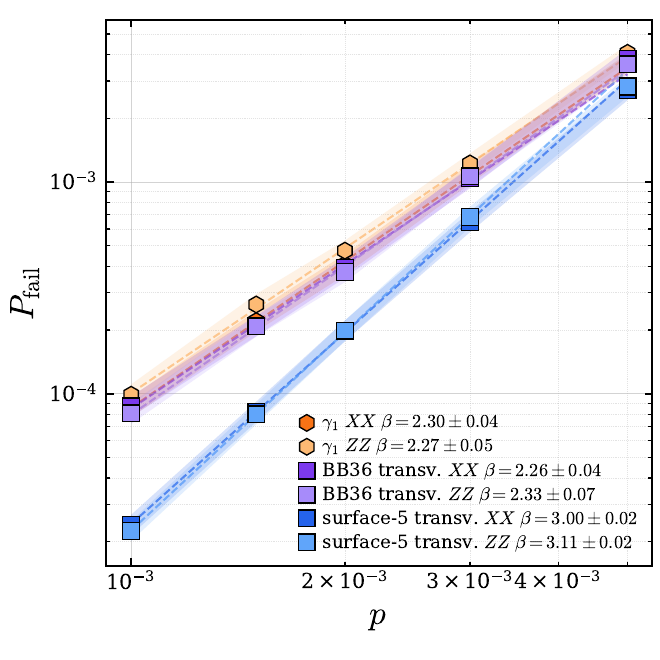}
        \caption{}
    \end{subfigure}
    \caption{\small \textbf{Simulation of a homomorphic CNOT via chain maps.} (a) Logical error rate for a memory experiment, with code $A$ and code $B$ idling for $R = 2\max(d_A, d_B)$ rounds under circuit-level noise. Here, code $A$ is a distance-$5$ rotated surface code, and code $B$ the $[[36, 8, 4]]$ bivariate bicycle code in the polynomial family of Ref.~\cite{Bravyi_2024}. (b) A depth-$3$ distance-preserving homomorphic $\overline{\mathrm{CNOT}}_{A_0\to B_0}$ from code $A$ to code $B$, compared with transversal $\overline{\mathrm{CNOT}}$s in the respective codes. The full experiment consists of $4$ logical $\overline{\mathrm{CNOT}}$ gates interleaved with QEC cycles. Dashed lines indicate power-law fits, and the fitted exponent $\beta$ characterises the effective scaling of the logical error rate; shaded areas indicate $95\%$ Wilson confidence intervals.} 
    \label{fig:surface-5_to_36_BB}
\end{figure}

In Table~\ref{tab:circ_table}, we demonstrate the broad applicability of using chain maps to implement inter-code entangling gates by listing various non-trivial code pairings. For simplicity of presentation, the unflagged homomorphic CNOT circuits presented in the table are full rank, i.e. couple $r=\min (k_A, k_B)$ logical degrees of freedom independently. For each chain map $\gamma_1$, we report the depth $\Delta$ of the resulting physical CNOT circuit, and the total number $w$ of physical CNOT gates. The quoted circuit-level distances $d_X^{\rm gad}$ and $d_Z^{\rm gad}$ were verified using \texttt{Stim}~\cite{Gidney_2021}. Many of the codes appearing in the first column are sources of magic, yielding a variety of possible interfaces for heterogeneous architecture design. We also include examples in which the logical qubits are distributed across multiple code blocks, as described in Appendix~\ref{assec:multi-block-targeting}.

The chain map conditions guarantee that the physical unitary $U_{\gamma_1}$ in Eq.~\eqref{eq:physical-cnot} preserves the product codespace and implements the target logical CNOT action. However, they are not sufficient to guarantee fault-tolerance of the realised circuit. Once $U_{\gamma_1}$ is implemented as a finite-depth CNOT circuit, faults occurring at intermediate times may propagate through later CNOTs into correlated errors. In the non-transversal case, certain error patterns may reduce the effective distance of the encoded operation. It is therefore notable that, for many of the examples in Table~\ref{tab:circ_table}, numerous distance-preserving solutions can nevertheless be found. When such a circuit is readily found, it constitutes a direct alternative to generalised surgery approaches. We numerically simulate the performance of  homomorphic CNOT gates compared against idling-memory baselines. Fig.~\ref{fig:surface-5_to_36_BB} shows the case of a $\overline{\mathrm{CNOT}}_{A_0 \to B_0}$ logical gate from a distance-$5$ surface code to the $[[36,8,4]]$ bivariate bicycle defined using polynomials $a=x^3+y+y^2$ and $b=y^3+x+x^2$ using $l=m=3$ as per Ref.~\cite{Bravyi_2024}. For comparison, we consider a baseline where both code blocks idle for $R = 2\max(d_A,d_B)$ (here 10) rounds, with logical performance shown in Fig.~\ref{fig:surface-5_to_36_BB}(a). The logical CNOT benchmark in~Fig.~\ref{fig:surface-5_to_36_BB}(b) then consists of $R$ rounds of interleaved QEC and homomorphic CNOT. We used a standard circuit-level depolarising noise model: before each round, each data qubit undergoes single-qubit depolarisation with probability $p$; each physical CNOT is followed by two-qubit depolarisation with probability $p$; resets and measurements fail independently with probability $p$. The syndrome extraction circuit for the rotated surface-$5$ code is the ubiquitous interleaved schedule, and that of the BB code follows the approach of Ref.~\cite{Bravyi_2024}. We run simulated experiments in both $(X_A,X_B)$ and $(Z_A,Z_B)$ bases, and given $P_{A,\mathrm{any}}$ and $P_{B,\mathrm{any}}$ the probabilities that any logical observable fails in the corresponding block, we define the joint any-logical failure probability by
\begin{equation}
    P_{\mathrm{joint,any}}
      = 1-(1-P_{A,\mathrm{any}})(1-P_{B,\mathrm{any}}).
\end{equation}
The plotted quantity is then normalised per logical and per round as
\begin{equation}
    P_{\mathrm{joint}}
      = 1-(1-P_{\mathrm{joint,any}})^{1/(k_A+k_B)},
    \qquad
    P_{\mathrm{joint}}^{\mathrm{round}}
      = 1-(1-P_{\mathrm{joint}})^{1/R}.
\end{equation}
The same per-logical normalisation is used for the CNOT experiments, with $P_{\mathrm{any}}$ taken to be the any-logical failure probability of the full CNOT protocol. The resulting logical error rates show that the homomorphic map provides a direct heterogeneous entangling operation with performance on the scale of the relevant memory and transversal-CNOT references, despite acting between codes with different structure and rate.

\begin{figure}[t!]
    \centering
    \includegraphics[width=0.7\linewidth]{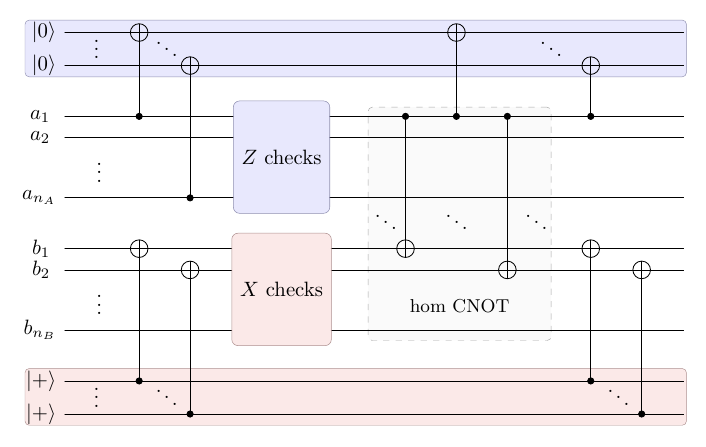}
    \caption{\small \textbf{Schematic circuit describing a flagged homomorphic CNOT between code $A$ and code $B$.} The blue and red boxes respectively describe $Z$ and $X$-stabilizer and flag measurements of codes $A$ and $B$. The logical homomorphic CNOT action indicated by the dashed box has arbitrary depth and gate count $n_{\mathrm{CNOT}}$.}
    \label{fig:flagged-hom-cnot}
\end{figure}

Without successful QEC-aware heuristics or formal guarantees, searching this space for a fault-tolerant solution is challenging. When such a solution is not easily found, the affine search space in Sec.~\ref{ssec:affine-restrictions} can still be used to find chain maps that give rise to partially distance-preserving homomorphic CNOT circuits. Such a $U_{\gamma_1}$ realised as a depth-$\Delta$ circuit can be made pieceably fault-tolerant~\cite{Yoder_2016} by inserting QEC cycles between different timesteps, so long as the intermediate codes obtained after every deformation layer have distance at least $\min(d_A, d_B)$. A particularly appealing approach is to start from a partially distance-preserving solution, and supplement the circuit with additional stabilizer and flag measurements to recover the full protocol distance. We show a schematic for such a procedure in Fig.~\ref{fig:flagged-hom-cnot}. Here, the role of the additional measurements is to detect harmful Pauli-$X$ errors on code $A$ and Pauli-$Z$ on code $B$ propagating between the CNOT layers. While this incurs additional measurement ancilla qubits, we emphasise that the flag measurements are single-shot in general, and empirically observe that the extra stabilizer measurements only require a single measurement round, when starting from a sufficiently fault-tolerant solution, leading to a favourable time overhead. These measurement outcomes are then incorporated into decoding during the homomorphic CNOT. In our examples we find a small and sufficient number of flags using a greedy pruning approach, and leave the design of adequate and scalable flagged circuits for future work~\cite{flag-cnot-scheme}. The flagged homomorphic CNOT examples in Table~\ref{tab:circ_table} show that a modest number of flags can be sufficient to restore the target circuit-level distance for otherwise non-fault tolerant circuits. This behaviour is numerically demonstrated in Fig.~\ref{fig:surface-7-to-BB72}, for a homomorphic CNOT from a distance-$7$ surface code to the $[[72,12,6]]$ bivariate bicycle code. Since the goal is to assess the performance of the homomorphic CNOT gate itself, for simplicity the numerical benchmarks henceforth use a circuit-level noise model tailored to the MPP (multi-Pauli product) syndrome-extraction circuits used throughout these simulations, using \texttt{Stim}. At the start of each QEC round and CNOT gate, data qubits undergo single-qubit depolarising noise with probability $p$. Each physical CNOT in a homomorphic or transversal CNOT layer is followed by two-qubit depolarising noise with probability $p$. Resets and measurements fail independently with probability $p$, and we set idle noise to zero. Stabilizer checks are implemented as MPP measurements whereby the Pauli product itself is measured ideally, while the reported syndrome bit is flipped with probability $p$. Here, the unflagged circuit has asymmetric circuit distance $(d_X,d_Z)=(4,6)$, while the pruned flagged construction restores distance $6$ using only five flags. We find that compared to the idling baseline in Fig.~\ref{fig:surface-7-to-BB72}(a), the flagged $\overline{\mathrm{CNOT}}_{A_0\to B_0}$ protocol improves the low-error scaling relative to the unflagged map, and performs comparably to the transversal CNOT on two $[[72,12,6]]$ blocks. 

\begin{figure}[t!]
    \centering
    \begin{subfigure}[b]{0.45\textwidth}
        \centering
        \includegraphics[width=\textwidth]{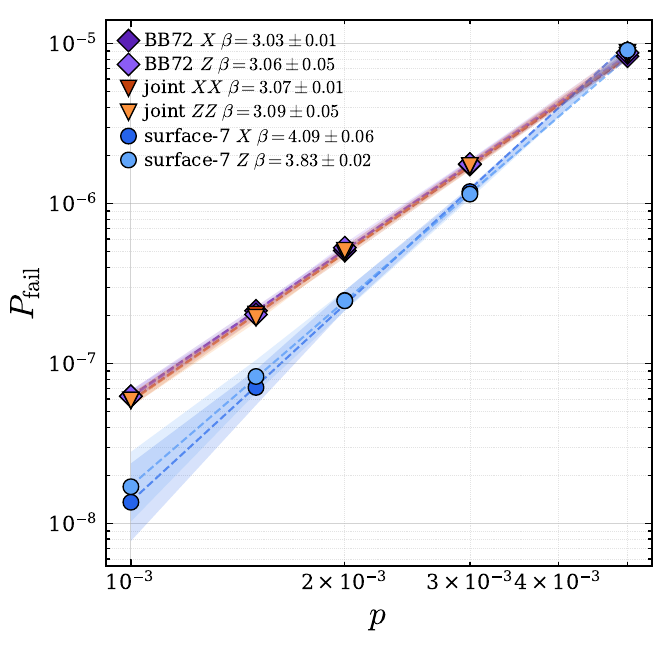}
        \caption{}
    \end{subfigure}
    \hfill
    \begin{subfigure}[b]{0.45\textwidth}
        \centering
        \includegraphics[width=\textwidth]{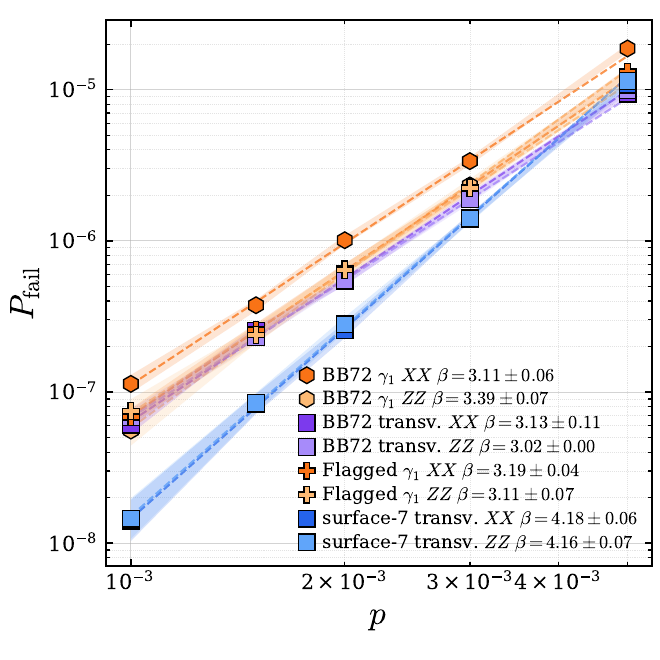}
        \caption{}
    \end{subfigure}
    \caption{\small \textbf{Simulation of a flagged chain map.} (a) Logical error rate for a memory experiment, with code $A$ and code $B$ idling for $R = 2\max(d_A, d_B)$ rounds under circuit-level noise. Here, code $A$ is a distance-$7$ rotated surface code, and code $B$ the $[[72, 12, 6]]$ bivariate bicycle code in the polynomial family of Ref.~\cite{Bravyi_2024}. (b) A depth-$4$ distance-preserving homomorphic $\overline{\mathrm{CNOT}}_{A_1\to B_1}$ from code $A$ to code $B$, compared with transversal $\overline{\mathrm{CNOT}}$s in the respective codes. The full experiment consists of $6$ logical $\overline{\mathrm{CNOT}}$ gates interleaved with QEC cycles.  The unflagged map has distance $d_X=4$ and $d_Z=6$, while the flagged map recovers the full code distance. Dashed lines indicate power-law fits, and the fitted exponent $\beta$ characterises the effective scaling of the logical error rate; shaded areas indicate $95\%$ Wilson confidence intervals.}
    \label{fig:surface-7-to-BB72}
\end{figure}

\subsubsection{Targeted inter-block CNOT gates}
\label{sec:targeted}

As a closely related application, simply performing inter-block CNOT gates (in a heterogeneous or homogeneous architecture) is a valuable computational primitive. However, implementing anything other than a depth-1 transversal CNOT between paired logical qubits of two identical code blocks is comparatively much more difficult. In general, targeted inter-block CNOT gates can be implemented again using universal adapters, bringing with it the aforementioned spacetime overheads. Certain codes admit bespoke methods for implementing arbitrary inter-block CNOT gates that outperform adapters. Here, we show that chain maps can outperform these bespoke methods as well.

One such example is the ConiQ~\cite{liu2025} compilation framework for the many-hypercubes code~\cite{Goto_2024}. Here we focus on one instruction of their compilation approach which they call the automorphism-assisted hierarchical addressing (AHA) scheme. The general idea of this scheme is that intra-block automorphisms implementing logical CNOT gates, when combined with transversal CNOT gates, are sufficient to generate arbitrary addressable logical CNOT gates between code blocks. We can search for bespoke chain maps implementing these addressable inter-block CNOT gates directly. Let us explicitly consider the $[[16,4,4]]$ code: using the ConiQ AHA scheme, targeted inter-block CNOT gates require an ancilla $[[16,4,4]]$ code block, and depth-$6$, see Fig.~7 of Ref.~\cite{liu2025}. Searching for bespoke chain maps gives depth-2 circuits for any rank-1 $\Gamma_Z(\gamma_1)$ between two $[[16,4,4]]$ codes that are fault-tolerant to distance-4. From this we obtain a $4.5\times$ reduction in spacetime overhead, without even considering the required error correction. When we include Steane-style QEC after every transversal CNOT in the AHA scheme, compared to a single round of Steane-style QEC after the chain map, we obtain greater than a $20\times$ reduction in spacetime volume. For the $[[64,8,8]]$ many-hypercubes code, we find a depth-4 targeted inter-block CNOT gate that is fault-tolerant to distance ($d_X=6, d_Z=8$). Fig.~\ref{fig:mhc_to_self} presents numerical simulations of a chain map implementing $\overline{\textrm{CNOT}}_{A_1 \rightarrow B_1}$ between two codes $A$, $B$, both of which are the $[[64,8,8]]$ MHC code.

A similar example is provided by phantom codes~\cite{koh2026}, on which arbitrary inter-block CNOT circuits between identical code blocks can be implemented with two layers of transversal CNOT gates and automorphisms. Two instances of phantom codes are the $[[16,3,4]]$ quantum Reed-Muller code and the $[[20,2,6]]$ CSD code~\cite{berthusen2025}. For both codes, we find that all arbitrary inter-block, unidirectional CNOT circuits can be implemented by a distance-preserving, depth-2 chain map. This matches the depth of the method proposed in Ref.~\cite{koh2026}; however, the chain maps only require a single round of error correction, whereas two rounds are required for the `phantom' method---one after each transversal CNOT. For larger phantom codes, the depth of the chain map may be larger than two. Depending on the cost of error correction and the number of flags required to make it fault-tolerant, chain maps may still be lower overhead.

\begin{figure}[t!]
    \centering
    \begin{subfigure}[b]{0.45\textwidth}
        \centering
        \includegraphics[width=\textwidth]{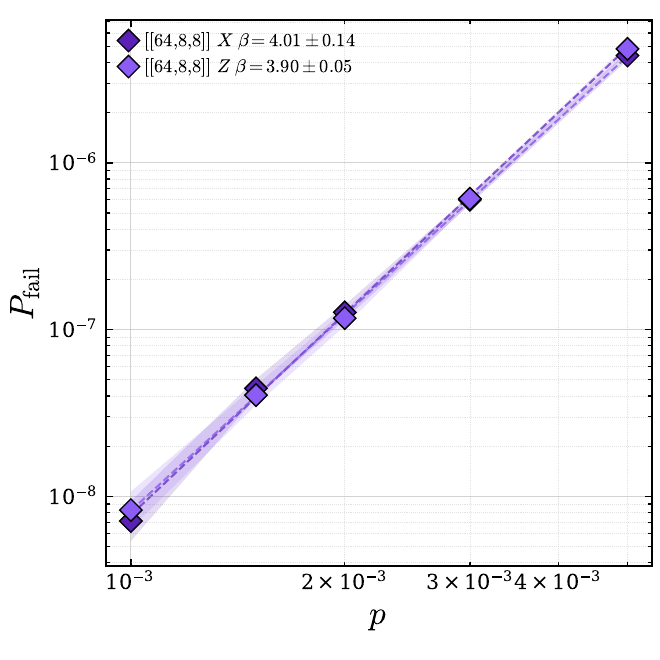}
        \caption{}
    \end{subfigure}
    \hfill
    \begin{subfigure}[b]{0.45\textwidth}
        \centering
        \includegraphics[width=\textwidth]{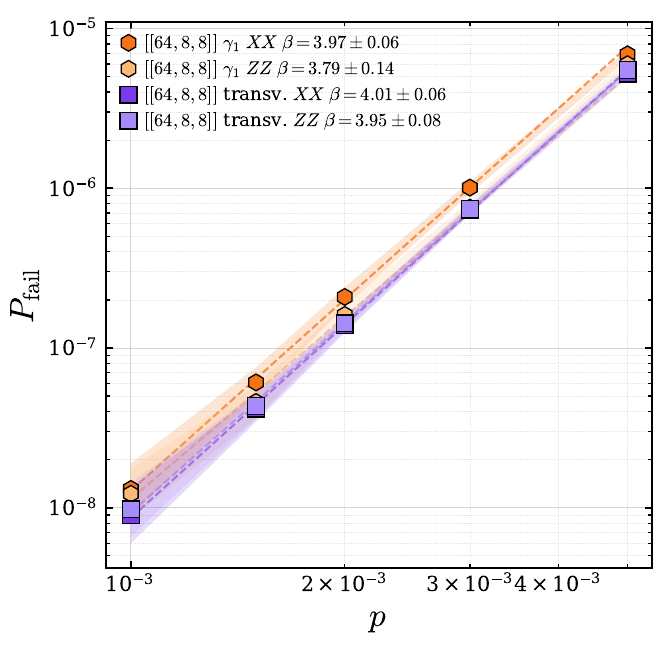}
        \caption{}
    \end{subfigure}
    \caption{\small \textbf{Simulation of a higher-distance chain map.} (a) Logical error rate for an idling memory experiment, where both codes $A$ and $B$ are $[[64,8,8]]$ many-hypercubes codes of Ref.~\cite{Goto_2024}, idling for $R = 2d$ rounds under circuit-level noise. (b) A depth-$3$ circuit realising a $\overline{\mathrm{CNOT}}_{A_1\to B_1}$ gate between the two $[[64,8,8]]$ blocks. The unflagged map has distance $d_X=6$ and $d_Z=8$, while the flagged map recovers the code distance. Dashed lines indicate power-law fits, and the fitted exponent $\beta$ characterises the effective scaling of the logical error rate; shaded areas indicate $95\%$ Wilson confidence intervals.}
    \label{fig:mhc_to_self}
\end{figure}

\subsection{Pauli product measurements}

Pauli product measurements (PPMs) and the underlying computational model of Pauli-based computation (PBC)~\cite{Bravyi_2016} have received significant attention since the rise in popularity of quantum low-density parity-check (qLDPC) codes. While qLDPC codes offer advantages over topological codes in terms of rate and distance, computing on them is relatively difficult owing to their high-density information storage. One approach that has recently seen substantial development is qLDPC code surgery~\cite{Cohen_2022, Cowtan_2024, Ide_2025, he2025, cowtan2026}, which facilitates PBC by enabling measurements of arbitrary logical Pauli operators. Universal computation across many logical qubits is achieved by introducing a source of magic and performing PPMs over several code blocks connected by bridges.

Generalised lattice surgery on qLDPC codes can be expressed in the framework of homological algebra, see Refs.~\cite{chang2026, Ide_2025}. Concretely, suppose we have a CSS code represented by a chain complex $B_\bullet$, on which we desire to measure a logical operator. Surgery introduces an ancilla system $A_\bullet$ and a chain map $\boldsymbol{\gamma} : A_\bullet \xrightarrow{} B_\bullet$, where $\boldsymbol{\gamma}  = (\gamma_1, \gamma_0)$ such that the following diagram commutes.
\begin{equation}
    \begin{tikzcd}
  {} 
    & C_1^A \arrow[r, "\partial_1^A"] \arrow[d, "\gamma_1"'] 
    & C_0^A \arrow[r, "\partial_0^A"] \arrow[d, "\gamma_0"'] 
    & C_{-1}^A  \\
  C_2^B \arrow[r, "\partial_2^B"'] 
    & C_1^B \arrow[r, "\partial_1^B"']
    & C_0^B 
    & {}
\end{tikzcd}
\label{eq:surgery}
\end{equation}
By considering Eq.~\eqref{eq:surgery} as a CSS code itself 
\begin{equation}
    AB_\bullet: C_1^A\oplus C_2^B \xrightarrow{\;\partial_1^{AB}\;} C_0^A\oplus C_1^B \xrightarrow{\;\partial_0^{AB}\;} C_{-1}^A\oplus C_0^B
\end{equation}
and measuring the stabilizer generators of this joint code-ancilla system, the desired logical operators of $B$ can be obtained. Surgery on qLDPC codes as presented here requires $O(d)$ time overhead and a number of ancilla qubits scaling with the weight of the logical to be measured. Recent advances have improved on the time overhead by leveraging redundancies in the ancilla systems~\cite{williamson2024, chang2026}.

\begin{figure}[t]
    \centering
    \includegraphics[width=\linewidth]{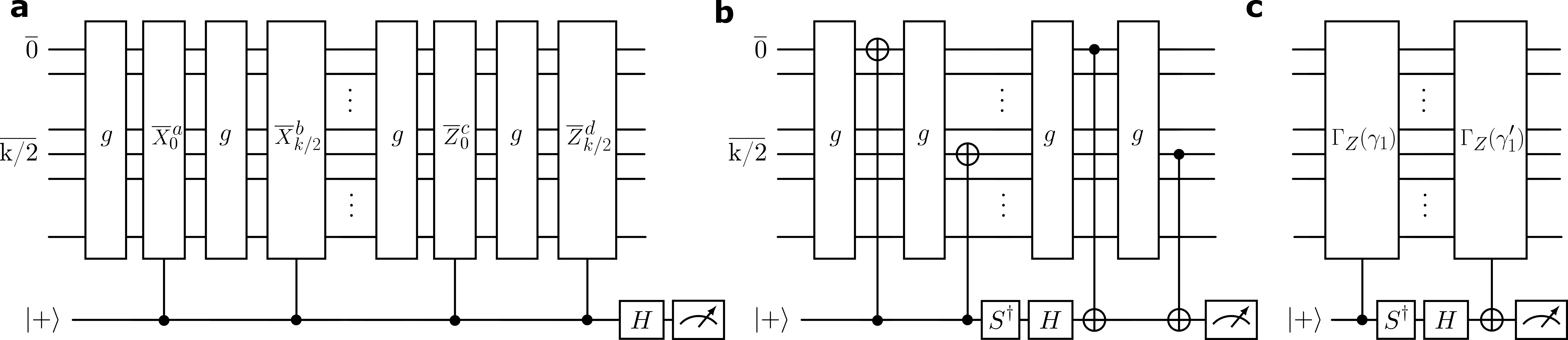}
    \caption{\small \textbf{Gadgets to perform PPMs via chain maps. a.} Circuit to (potentially non-fault-tolerantly) measure the seed operators $\overline{X}_0$ and $\overline{X}_{k/2}$ ($Z$) using some logical ancilla. \textbf{b.} CZ gates can be replaced with CX gates and the appropriate change of basis on the ancilla. \textbf{c.} $X/Z$ can be combined into a single $\Gamma_Z(\gamma_1)$ operation to reduce depth. }
    \label{fig:ppms}
\end{figure}

We focus on one specific instantiation of qLDPC code surgery as presented in Refs.~\cite{webster2025, webster2026}. There, the authors investigate surgery gadgets for a family of generalised bicycle (GB) codes~\cite{kovalev2013} which have the convenient property that a small number of logical \textit{seed operators}, when combined with automorphisms of the code, are sufficient to construct any logical operator of the code. Consequently, only a small number of surgery gadgets $A^i_\bullet$, and adapter/bridge systems to connect them, are required to measure an arbitrary Pauli product. The required overheads are presented in Table~1 of Ref.~\cite{webster2026}, and we note here that the $[[30,8,4]]$ and $[[62,10,6]]$ GB codes require 80 and 120 ancilla qubits, respectively, to construct the four required surgery gadgets, and $O(d)$ depth to measure the desired Pauli product fault-tolerantly.

\begin{figure}[!t]
    \centering
    \begin{subfigure}[b]{0.45\textwidth}
        \centering
        \includegraphics[width=\textwidth]{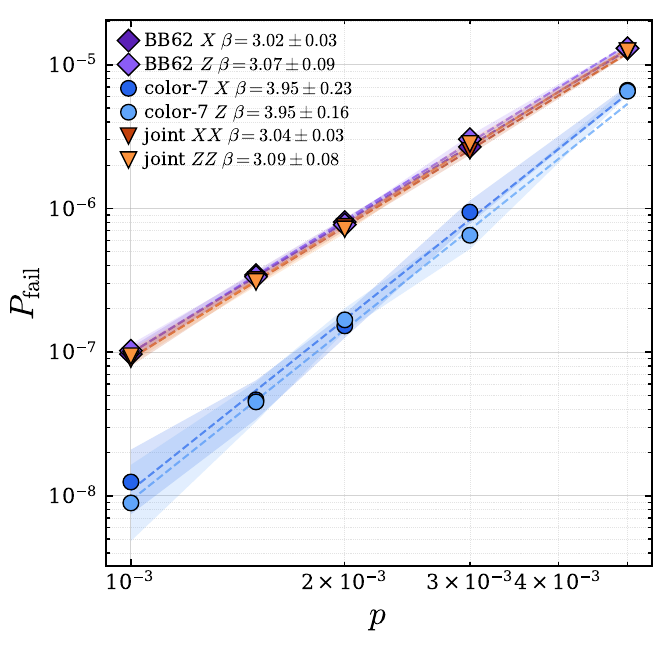}
        \caption{}
    \end{subfigure}
    \hfill 
    \begin{subfigure}[b]{0.45\textwidth}
        \centering
        \includegraphics[width=\textwidth]{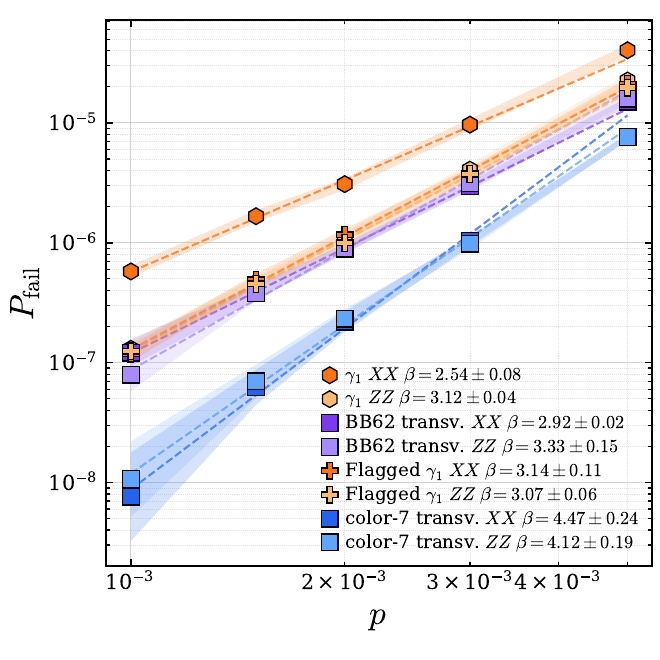}
        \caption{}
    \end{subfigure}
    \caption{\small \textbf{Simulation of a PPM-style circuit via chain maps.} (a) Logical error rate for an idling memory experiment, with code $A$ and code $B$ idling for $R = 2\max(d_A, d_B)$ rounds under circuit-level noise. Here, code $A$ is a distance-7 6.6.6 color code, and code $B$ is the $[[62,10,6]]$ generalised bicycle (GB) code~\cite{webster2025}. (b) Logical error rate for a logical FANOUT:  $\overline{\mathrm{CNOT}}_{A_0\to B_0} \overline{\mathrm{CNOT}}_{A_0\to B_5}$, from the 6.6.6 color-$7$ to the $[[62,10,6]]$ generalised bicycle code from Ref.~\cite{webster2025}. A fault-tolerant $(d_X = d_Z= 6)$ circuit requiring 18 flags is compared with an unflagged $(d_X=4,d_Z=5)$ circuit as well as transversal $\overline{\mathrm{CNOT}}$s in the respective codes. The full experiment consists of $d=6$ logical $\overline{\mathrm{CNOT}}$ gates interleaved with a QEC cycle. Dashed lines indicate power-law fits, and the fitted exponent $\beta$ characterises the effective scaling of the logical error rate; shaded areas indicate $95\%$ Wilson confidence intervals.
    }
    \label{fig:color7_to_BB62}
\end{figure}

Adapting this approach of measuring seed operators to instead use logical CNOT and CZ gates yields the circuits in Fig.~\ref{fig:ppms}(a),(b), where $g$ is some automorphism of the code, and $a,b,c,d\in\{0,1\}$, i.e., choosing to measure the operator or not. Note that an $S^\dagger$ gate on the ancilla may be needed to account for the phase in the controlled-$iXZ$ gate. We find fault-tolerant, depth-3 chain maps from the $[[30,8,4]]$ GB code to the $[[17,1,5]]$ 4.8.8 color code~\cite{Bombin_2006} for each of the four required $\Gamma_Z(\gamma_1)$'s, hence yielding a depth-12 circuit to measure an arbitrary Pauli $P \in \mathcal{\overline{P}}_8$ in the worst case. We obtain similar overheads when using the $[[37,1,7]]$ 6.6.6 color code to measure Pauli products on the larger $[[62,10,6]]$ GB code. Performing the logical FANOUT gate $CNOT_{A_0\to B_0} CNOT_{A_0\to B_5}$ can be done in depth-5; however, 18 flags are then required for it to be fault-tolerant to distance 6. Fig.~\ref{fig:color7_to_BB62} presents circuit-level simulations of this flagged FANOUT gate.

The depth may be reduced whenever $\overline{X}_0$ and $\overline{X}_{k/2}$ ($Z$) are measured following the same automorphism $g$, or if the desired Pauli is purely $X/Z$-type. To make the entire logical measurement fault-tolerant, it may be required to perform syndrome extraction between logical CNOTs, or to introduce a logical ancilla as a flag. Therefore, $O(d)$ time will also be required when using chain maps; however, we have reduced the number of ancilla qubits by $\sim 2$--$4\times$. Given the flexibility we have in designing the logical operator $\Gamma_Z(\gamma_1)$, we do not have to restrict ourselves to measuring $\overline{X}_0$ and $\overline{X}_{k/2}$ ($Z$) combined with automorphisms. Instead, we can consider applying the circuit in Fig.~\ref{fig:ppms}(c) in which the $X$- and $Z$-type operators are measured with a single $\Gamma_Z(\gamma_1)$. Sampling random FANIN/OUT between the $[[30,8,4]]$ GB code and the $[[17,1,5]]$ color code yields depth-3 chain maps that are fault-tolerant to distance-4. Most generally, we could consider replacing the self-dual $[[n,1,d]]$ with a code containing $k > 1$ logical qubits. Then, $k$ commuting Pauli products could be measured in parallel by choosing the appropriate $\Gamma_Z(\gamma_1) \in \mathbb{F}_2^{k \times k'}$.

\subsubsection{Grid Pauli product measurements}

One notable instantiation of a PPM scheme in which ancillary surgery systems are not required is the Grid Pauli product measurement (GPPM)~\cite{xu2025fast}. GPPMs are gadgets for hypergraph product (HGP) codes~\cite{Tillich_2014} that can measure Pauli products on any subgrid of the HGP logical qubits. This is facilitated by a degree-1 chain map $\gamma_1 \in \mathrm{Hom}_1(\mathcal{Q}', \mathcal{Q})$ between a data code $\mathcal{Q}$ and a suitably constructed ancilla code $\mathcal{Q}'$, whose associated physical unitary $U_{\gamma_1}$ implements a homomorphic CNOT.

An HGP code $\mathcal{Q} = \textrm{HGP}(H_1, H_2)$ is constructed from the tensor product of two length-2 classical chain complexes, $A_{\bullet}$ and $B_{\bullet}$ representing classical codes $[n_1,k_1,d_1], [n_2,k_2,d_2]$,
\begin{equation}
    A_\bullet: C_1^A \xrightarrow{\;\partial_1^A\;} C_0^A, \qquad B_{\bullet}: C_1^B \xrightarrow{\;\partial_1^B\;} C_0^B,
\end{equation}
with boundary maps given by the parity-check matrices $H_1 = \partial_1^A$, $H_2 = \partial_1^B$. The logical qubits of $\mathcal{Q}$ inherit a $[k_1] \times [k_2]$ grid structure arising from the tensor product. The ancilla code $\mathcal{Q}'$ is obtained by modifying the base classical complexes via puncturing or augmenting. 

Both puncturing and augmenting induce chain maps from the modified and original length-$2$ chain complexes. Concretely, if $H'$ is obtained from $H$ by puncturing or augmenting, then we have the valid chain map $\boldsymbol{\gamma} : H'_\bullet \xrightarrow{} H_\bullet$ with $\boldsymbol{\gamma} = (\gamma_1: C_1^{H'} \rightarrow C_1^{H}, \gamma_0: C_0^{H'} \rightarrow C_0^{H})$. This `classical' chain map lifts to a chain map $\gamma_1 \in \mathrm{Hom}_1(\mathcal{Q}', \mathcal{Q})$ between the ancilla code $\mathcal{Q}' = \textrm{HGP}(H', H')$ and the data code $\mathcal{Q}$. The induced logical action $\Gamma_Z(\gamma_1) \in \mathbb{F}_2^{k \times k'}$ realises a full-rank homomorphic CNOT. Furthermore, the physical unitary $U_{\gamma_1}$ consists entirely of a single layer of physical CNOTs between matched physical qubits of the two code blocks. 

While GPPMs combined with transversal/fold-transversal gates are powerful enough to facilitate the implementation of the entire Clifford group on the $O(k^2)$ logical qubits of an arbitrary HGP code, they are somewhat limiting in the sense that products of Paulis arranged in a \textit{grid} are measurable. This results in certain measurements being difficult to perform. To that end, we can directly search for chain maps $\gamma_1$ implementing the desired logical action, without assuming any special structure on $\gamma_1$ beyond the chain map constraints themselves. As discussed previously, we can use logical $k=1$ ancilla to efficiently perform arbitrary PPMs. In order to maintain the ability to parallelise PPMs, and homogenise state preparation, it makes sense to utilise the high-rate data code $\mathcal{Q}$ as the ancilla block as well. 

\begin{figure}
    \centering
    \includegraphics[width=0.95\linewidth]{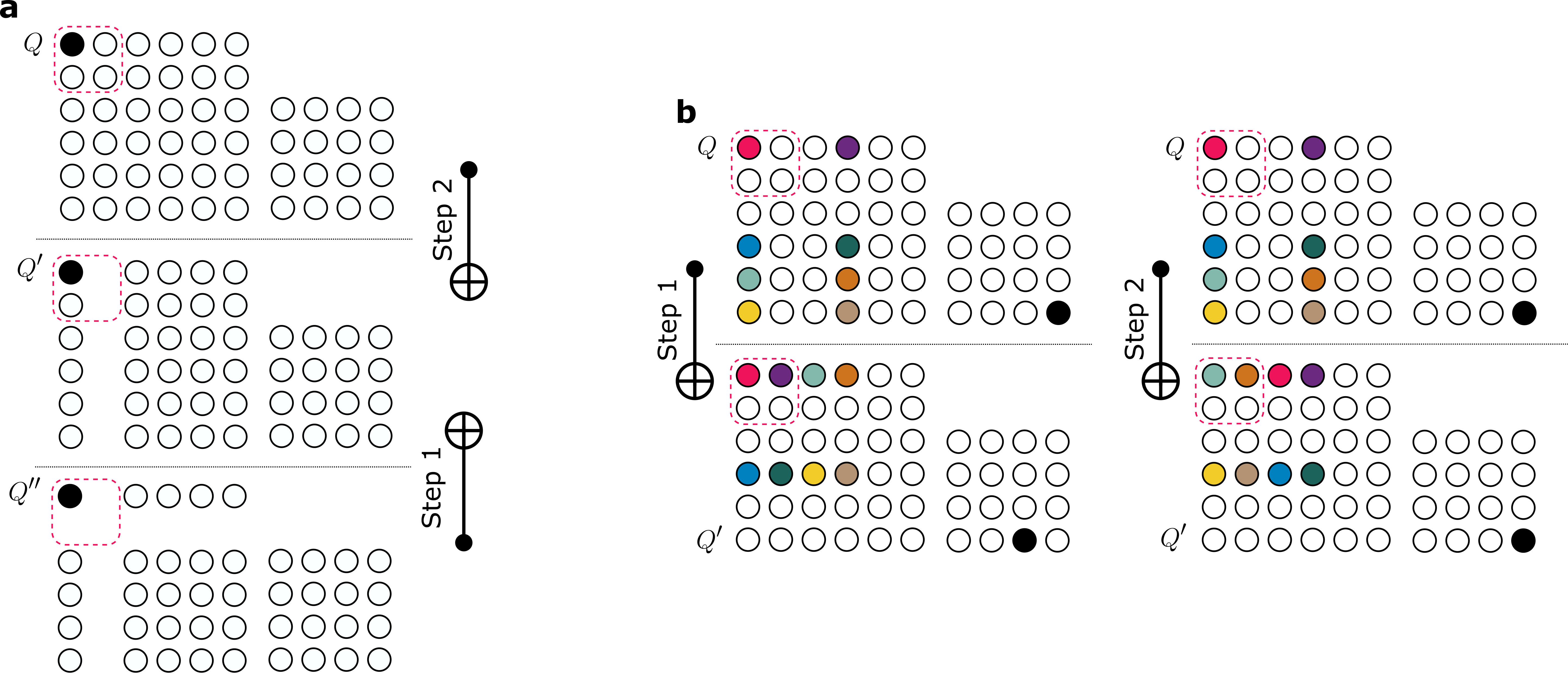}
    \caption{\small \textbf{Comparison of GPPMs and bespoke chain maps. a.} The procedure to measure the top-left (0,0) logical qubit using a GPPM. Two ancillary code blocks, $\mathcal{Q}'$ and $\mathcal{Q}''$ are required. \textbf{b.} Measuring the (0,0) logical qubit using a bespoke chain map. Only one ancilla code block is required, and the depth of the procedure is still two. Qubits highlighted with the same colour indicate that a physical CNOT gate is applied between them in that timestep.}
    \label{fig:gppms}
\end{figure}

As a concrete example, we consider a $[[52,4,4]]$ HGP code $\mathcal{Q}$ constructed from one-generator systematic-circulant (OGSC) base classical codes, see Appendix A of Ref.~\cite{xu2025fast}. Under the GPPM framework, measuring rows or columns of logical qubits in the $2 \times 2$ grid requires preparing a single ancilla block $\mathcal{Q}'$ with the appropriate rows or columns punctured. Measuring a single logical $Z$ operator on, for example, the top-left $(0,0)$ qubit in Fig.~\ref{fig:gppms}(a), requires a secondary ancilla $\mathcal{Q}''$ with additional rows and columns punctured. Instead, we can search for the targeted inter-block logical CNOT gate that would allow us to directly measure this logical operator using a single ancilla code block. This corresponds to imposing the affine constraint that only $\Gamma_Z(\gamma_1)_{00} = 1$, hence restricting to the subset $S_A = S_B = \{0\}$. We find the depth-2 chain map depicted in Fig.~\ref{fig:gppms}(b) which implements the desired logical action. In a given timestep, nine physical CNOT gates are applied to qubits highlighted with the same colour. We verify that $\gamma_1$ is fault-tolerant to distance-4 using \texttt{Stim}~\cite{Gidney_2021}. Compared to the GPPM implementation, we have removed the need for the second ancillary block $\mathcal{Q}''$, while maintaining the fault-tolerance and depth. We can also construct chain maps for measurements that are products of Pauli measurements, such as $\overline{Z}_{0,0} \overline{Z}_{0,1}\overline{Z}_{1,0}\overline{Z}_{1,1}$. Implementing this with logical CNOTs corresponds to a rank-1 logical FANIN into $\mathcal{Q}'$ from $\mathcal{Q}$. Concretely, we target the subset mapping $S_A = \{0,1,2,3\}, S_B=\{0\}$, and find a depth-3 chain map, $\gamma_1$, using 45 physical CNOT gates that is fault-tolerant to distance-4. 

Note that using bespoke chain maps is less efficient than GPPMs for certain measurements; for example, measuring the $\overline{Z}^{\otimes k}$ operator of a single column or row with a GPPM requires only $\mathcal{Q}'$ and a depth-1 circuit. Searching for the corresponding $\Gamma_Z(\gamma_1)$'s yielded maps with at least depth-2. When restricting to using $\mathcal{Q}' = \mathcal{Q}$ as an ancilla, the asymptotic overheads for the two schemes are identical, as $O(n)$ physical qubits are required for $\mathcal{Q}'$, and $O(d)$ time is required to fault-tolerantly prepare it. Nonetheless, removing the need for $\mathcal{Q}''$ while maintaining low-depth may have significant benefits in practice. As previously discussed, it may also be beneficial to use a smaller ancilla code, particularly for measurements of high-weight Pauli products. 

\subsection{Operations on concatenated codes}

Code concatenation~\cite{gottesman1997} is a procedure which takes two or more input stabilizer codes and outputs a larger stabilizer code. An instructive way to interpret the resulting code is as follows: consider two CSS codes $A$ and $B$ with parameters $[[n_a,1,d_a]]$ and $[[n_b,1,d_b]]$ (generalising code concatenation in cases when $A$ or $B$ has multiple logical qubits is straightforward, but the resulting distance of $A \circ B$ depends heavily on its exact structure), respectively, with which we will construct $A \circ B$, an $[[n_an_b,1,d_ad_b]]$ concatenated code. First, take $n_b$ disjoint instances of code $A$, $A^{\otimes n_b}$, which altogether can be viewed as a $[[n_an_b,n_b,d_a]]$ code. We then add the $n_b-1$ stabilizer generators coming from $B$ which reduce the size of the overall logical subspace to a single logical qubit. These additional stabilizer generators are unique as they act on the \textit{logical} qubits of $A^{\otimes n_b}$, i.e., from a stabilizer generator $P_iP_j\dots P_l$ of $B$, $\overline{P}_i \overline{P}_j\dots \overline{P}_l$, where $\overline{P}_i$ is a logical operator for the $i$-th $A$ code block, is included as a stabilizer generator of $A \circ B$. With this framing, measurement of these generators coming from $B$ can be interpreted as performing a Pauli product measurement on $A^{\otimes n_b}$. As such, we can perform syndrome extraction by constructing a chain map implementing a FANIN operation between $A^{\otimes n_b}$ and an ancilla code $\mathcal{Q}'$ with the appropriate distance. Multiple concatenated stabilizer generators could be measured in this way if we use an ancilla code $\mathcal{Q}'$ with more than one logical qubit. 

Logical operations on concatenated codes can also be interpreted hierarchically, where logical operations on the lower level code blocks enact logical operations on the overall code. As a concrete example, let us consider a two-level concatenated code in which the physical qubits of a clustered-cyclic (CC) code~\cite{gu2026qgpu} are encoded into a $[[n,n-2,2]]$ iceberg code~\cite{Self_2024}. CC codes have a clustered logical basis in which $\textrm{supp}(\overline{X}_i) ~\cap~ \textrm{supp}(\overline{X}_j) = \emptyset$ and $\textrm{supp}(\overline{X}_i)  = \textrm{supp}(\overline{Z}_i)$. Let us assign each iceberg codeblock to one of these clusters. Now consider performing a round-robin~\cite{Yoder_2016} $\overline{\textrm{CNOT}}_{i\rightarrow j}$ gate achieved by applying physical CNOT gates between the support of $\overline{X}_i$ and $\overline{Z}_j$ in a round-robin fashion. For a code with distance $d$, this operation would normally require $d$ rounds of physical CNOT gates to implement; however on the CC $\circ$ iceberg concatenated code, it can be done in constant time. In particular, the chain map implementing $\Gamma_Z(\gamma_1) = \mathbf{1}$, i.e. the all-ones matrix corresponding to a logical round-robin circuit, can be done in depth-2 for all iceberg codes. Note that the chain map is only fault-tolerant to distance two, and so additional machinery will be required to make it fully fault-tolerant on the concatenated code.

\section{Discussion}
\label{sec:discussion}

In this work, we introduced a linear-algebraic framework for synthesising inter-code logical Clifford couplings between CSS codes using chain maps. This relies on the observation that the physically relevant CNOT couplings between two code blocks $A$ and $B$ form a vector space, namely the degree-$1$ hom-space $\HomI{B}{A}$, and that specifying a logical action restricts this space to an affine subspace. This separates the problem into two conceptually distinct stages: first constructing the affine family of logically valid chain maps, and then searching within this family for representatives with favourable circuit properties. Since this construction depends only on the chain-complex data of the input codes, it is agnostic to geometry, locality, and code family, and therefore applies naturally to QEC architectures and primitives that exploit heterogeneous encodings.  

The examples in Sec.~\ref{sec:applications} show that this framework produces non-trivial low-depth solutions, including near-transversal behaviour between seemingly unrelated code pairs. These maps can realise logical CNOT networks of varying rank, logical FANOUT/IN primitives, and targeted subspace couplings. We described how such operations are relevant for a range of QEC applications, including code switching, magic state injection, Pauli product measurements, and operations on concatenated codes. Our numerical simulations further suggest that the resulting circuits can have performance comparable to transversal CNOT gates and idling memory. Beyond the examples considered here, the same framework can be used to design other primitives, such as GHZ-state preparation, bus-mediated phase rotation gadgets, and multi-block Clifford gadgets in which one code block mediates interactions between many others. For example, the latter can be used to distribute efficiently prepared multi-controlled non-Clifford resources~\cite{li2025} into heterogeneous and/or multiple code blocks via teleportation. 

The main computational bottleneck of our method is not the construction of the admissible chain-map space itself, but the subsequent optimisation over the resulting affine family. This optimisation is motivated by the transversal and near-transversal maps known in the literature~\cite{li2025, tan2025}. Transversal maps correspond to partial matchings, and are therefore straightforwardly schedulable once found, whereas depth minimisation over a general affine family is a discrete optimisation problem with no evident polynomial-time solution in general. In our implementation, this step is handled by CP-SAT. The size of the resulting search problem depends strongly on the input codes and on the logical target. The physical coupling matrix has $n_A n_B$ binary entries, and the admissible chain maps form a linear subspace characterised by the chain-map constraints. Imposing a desired logical action further restricts this to an affine subspace. Although these linear-algebraic steps are efficient compared with exhaustive search, the residual affine space can still have large dimension, so exhaustive brute-force enumeration is generally impractical. Consequently, brute-force search over the admissible family scales at worst exponentially in $n_A n_B$. Logical targeting removes at most $O(k_Ak_B)$ dimensions, and therefore does not substantially change this scaling under the code-rate assumptions considered here. The CP-SAT formulation provides a structured optimisation over this exponentially large affine space, but it also limits the scale of the instances accessible in our investigation. While the examples presented here should therefore be viewed as proof-of-principle demonstrations that the chain-map formulation can discover non-trivial logical interfaces, we highlight that we were able to discover circuits relevant to current moderate-size architecture proposals~\cite{Bravyi_2024, li2025, liu2025, webster2026}. 

Given this complexity, we introduced simple heuristics that empirically improve solver behaviour. These include promoting sparsity in the search-space basis, as described in Sec.~\ref{ssec:sparse-basis}, as well as a decongestion heuristic in place of direct depth-minimisation, which we generally found to reduce the time to first solution under a given maximum depth constraint. Given a current depth cap $D$, this heuristic asks how close a representative is to satisfying the stronger cap $D-1$, and penalises both the number of qubits whose row or column degree exceeds $D-1$, and the total excess above this threshold. If the overflow is reduced to zero, the map $\gamma_1$, and hence the physical CNOT circuit induced by $\gamma_1$, has been decongested to depth at most $D-1$, and the procedure is repeated. This replaces the coarse depth-minimisation objective by a smoother repair problem, redistributing CNOTs away from the overloaded qubits while remaining inside the affine family of logically valid maps.

Nevertheless, circuit depth is a coarse proxy for fault-tolerance, as it bounds the spread of errors in time, but does not by itself capture hook-error mechanisms and correlated fault paths. Future implementations should therefore incorporate QEC-aware objectives directly into the optimisation. Possible avenues include penalties for concentrating CNOTs on representatives of the same logical class, imposing constraints derived from code structure, and exploiting symmetry reductions from code automorphisms. It would also be interesting to investigate whether valid solutions between small instances of two code families can be lifted to ans\"atze for the larger instances. More broadly, previous solver outputs could be used to learn promising masks, basis changes, branching priorities, or other structural features that make the affine search space easier to navigate. 

The formalism developed here is also not limited to the specific family-builder and logical restrictor routines used in this work. Indeed, the framework provides a reusable algebraic interface for designing logical gadgets beyond the particular instances studied here, including diagonal Clifford gates, as shown in Appendix~\ref{asec:cz-extension}, larger Clifford primitives, as well as non-Clifford logic. We leave such generalisations for future work.  
 
Finally, we emphasise that the affine space realising a desired logical action already contains many partially fault-tolerant candidates. Hence, any map in this space provides a circuit which can in principle be promoted to a fault-tolerant inter-code coupling using  flags~\cite{flag-cnot-scheme}. We leave the spacetime tradeoff analysis of such schemes in comparison with generalised surgery gadgets~\cite{swaroop2026}, as well as applications to surgery schemes, for future work.

\section{Acknowledgements}

We thank \cb{S}erban Cercelescu for discussions on flagging Clifford circuits. We also thank Jonas Anderson and Mateusz Kupper for reviewing our manuscript. 

\printbibliography

\appendix

\section{Supplementary material}

\subsection{Degree-1 chain-map dimension}
\label{assec:gamma1-dim}

\begin{proposition}
    Choose decompositions 
    \begin{equation}
    C_1^A = B_1^A \oplus L^A \oplus T^A,
    \qquad
    C_1^B = B_1^B \oplus L^B \oplus T^B,
    \end{equation}
    with $B_1 = \mathrm{im}\partial_2$ and $Z_1 = \ker \partial_1 = B_1 \oplus L$. Then, $\gamma_1\in\operatorname{Hom}_1(B, A)$ if and only if, in the corresponding block coordinates, 
    \begin{equation}
        \gamma_1 =
        \begin{pmatrix}
             \gamma_{1,BB} & \gamma_{1,BL} & \gamma_{1,BT} \\
             0 & \gamma_{1,LL} & \gamma_{1,LT} \\
             0 & 0 & \gamma_{1,TT} \\
        \end{pmatrix},
        \label{eq:gamma1_block_2}
    \end{equation}
    where rows are indexed by $B_1^A, L^A, T^A$ and columns by $B_1^B, L^B$ and $T^B$. Therefore,
    \begin{equation}
        \dim \operatorname{Hom}_1(B,A) = r_Z^Ar_Z^B + k_B(r_Z^A + k_A) + r_X^Bn_A,
    \end{equation}
    where $r_Z^Q:=\dim\operatorname{im}\partial_2^Q$ and $r_X^Q:=\operatorname{rank}\partial_1^Q$ denote check ranks.
\end{proposition}
\begin{proof}
    By Proposition~\ref{prop:hom1}, 
    \begin{equation}
    \gamma_1(B_1^B)\subseteq B_1^A,
    \qquad
    \gamma_1(B_1^B\oplus L^B)\subseteq B_1^A\oplus L^A.
    \end{equation}
    The first inclusion implies that the components $\gamma_{1,LB}$ and $\gamma_{1,TB}$ vanish, while the second inclusion implies that $\gamma_{1,TL}$ vanishes. Conversely, if these three blocks vanish, we recover exactly the condition in Proposition~\ref{prop:hom1}. Using $\dim B_1^A = r_Z^A$, $\dim L^A = k_A$, $\dim T^A = r_X^A$, $n_A = r_Z^A + k_A + r_X^A$, and similarly for code $B$, the dimension formula follows. 
\end{proof}

\subsection{Multi-block logical action targeting}
\label{assec:multi-block-targeting}

The logical-action formalism presented in Secs.~\ref{sec:logical-action} and~\ref{sec:logical-targeting} extends directly to collections of CSS code blocks. Instead of coupling a single pair of codes, we consider a set of codes $\{A_\ell\}_{\ell=1}^L$, $\{B_m\}_{m=1}^M$, and form their degree-wise direct sums  
\begin{equation}
A := \bigoplus_{l=1}^L A_l,
\qquad
B := \bigoplus_{m=1}^M B_m.
\end{equation}
At the chain-complex level, this is described by $C_i^A=\bigoplus_{\ell=1}^L C_i^{A_\ell}$ and $
C_i^B=\bigoplus_{m=1}^M C_i^{B_m}$,
with boundary operators
\begin{equation}
    \partial_i^A = \bigoplus_l \partial_i^{A_l}, \qquad \partial_i^B = \bigoplus_m \partial_i^{B_m}.
\end{equation}
Since $A$ and $B$ are again CSS chain complexes, all the constructions of Sec.~\ref{sec:build-hom-space}--\ref{sec:logical-targeting} apply without change. Let 
\[
n_A:=\sum_{\ell=1}^L n_{A_\ell},
\qquad
n_B:=\sum_{m=1}^M n_{B_m}.
\]
Then, using the direct-sum decompositions above, a degree-$1$ map $\gamma_1 \in \operatorname{Hom}_1(B, A)$ is represented by a matrix 
\begin{equation}
\gamma_1=
\begin{pmatrix}
\gamma_1^{(1,1)} & \cdots & \gamma_1^{(1,M)} \\
\vdots & \ddots & \vdots \\
\gamma_1^{(L,1)} & \cdots & \gamma_1^{(L,M)}
\end{pmatrix},
\qquad
\gamma_1^{(\ell,m)}\in \mathbb F_2^{n_{A_\ell}\times n_{B_m}}.
\label{eq:gamma1-block}
\end{equation}
Each block $\gamma_1^{(\ell,m)}$ specifies the CNOT pattern from block $A_l$ (control)
to block $B_m$ (target) in the physical unitary $U_{\gamma_1}$ as per Eq.~\eqref{eq:physical-cnot}. Let $k_{A_l}$ and $k_{B_m}$ denote the number of logical qubits in the respective subcodes
\[
k_A:=\sum_{\ell=1}^L k_{A_\ell},
\qquad
k_B:=\sum_{m=1}^M k_{B_m},
\]
and choose logical representative matrices compatible with the block decomposition, for example
\[
\boldsymbol L_Z^A=\operatorname{diag}(\boldsymbol L_Z^{A_1},\dots, \boldsymbol L_Z^{A_L}),
\qquad
\boldsymbol L_Z^B=\operatorname{diag}( \boldsymbol L_Z^{B_1},\dots, \boldsymbol L_Z^{B_M}),
\]
and similarly for $\boldsymbol L_X^A$ and $\boldsymbol  L_X^B$. Then the induced logical matrix $\Gamma_Z \in \mathbb{F}_2^{k_A\times k_B}$ also decomposes into blocks
\begin{equation}
\Gamma_Z(\gamma_1)=
\begin{pmatrix}
\Gamma_Z^{(1,1)} & \cdots & \Gamma_Z^{(1,M)} \\
\vdots & \ddots & \vdots \\
\Gamma_Z^{(L,1)} & \cdots & \Gamma_Z^{(L,M)}
\end{pmatrix},
\qquad
\Gamma_Z^{(\ell,m)}\in \mathbb F_2^{k_{A_\ell}\times k_{B_m}}.
\label{eq:GammaZ-block}
\end{equation}
Each block $\Gamma_Z^{(\ell,m)}$ captures the induced propagation of logical $Z$-operators
from block $B_m$ into block $A_\ell$. The same discussion applies to $\Gamma_X(\gamma_1)$, and in symplectically normalised logical bases one has $\Gamma_X(\gamma_1)=\Gamma_Z(\gamma_1)^{\mathsf T}$.

The affine targeting of Sec.~\ref{sec:logical-targeting} extends straightforwardly to this setting. Indeed, after choosing block-compatible logical representatives, the logical action map remains linear. Consequently, any linear condition on the block-structured logical action defines an affine subspace of $\HomI{B}{A}$. This includes prescribing a target submatrix $\Gamma_Z^{(\ell,m)}$ on chosen block pairs, imposing spectator behaviour on selected logical blocks or logical subspaces, and combining several such conditions simultaneously. By contrast, structured conditions involving rank or blockwise CNOT/FANOUT/IN conditions in which the source or target logical subspaces are not fixed in advance, are nonlinear or combinatorial, and can be treated by the same strategies described in Sec.~\ref{ssec:nonlinear-targeting}.

\section{Logical CZ from chain maps}
\label{asec:cz-extension}

\subsection{Formalism}

As seen in the main text, chain maps are naturally matched to CNOT-type propagation, since they preserve CSS operators at the chain-complex level, a point which was also mentioned in Ref.~\cite{poirson2025engineering}. We now describe the analogous formulation for inter-code CZ couplings. The main difference is that a CZ gate is diagonal and symmetric, i.e. it fixes all $Z$ operators, while propagating $X$ operators on the opposite block. Let 
\begin{equation}
    A_\bullet:
    C_2^A \xrightarrow{\partial_2^A}
    C_1^A \xrightarrow{\partial_1^A}
    C_0^A,
    \qquad
    B_\bullet:
    C_2^B \xrightarrow{\partial_2^B}
    C_1^B \xrightarrow{\partial_1^B}
    C_0^B
\end{equation}
be the chain complexes associated with two CSS codes, with $\partial_2 = H_Z^{\mathsf T}$, and $\partial_1 = H_X$. A physical CZ edge between qubits $a$ of $A$ and $b$ of $B$ acts by 
\begin{equation}
    X_a \mapsto X_aZ_b, \qquad X_b \mapsto Z_aX_b,
\end{equation}
while fixing $Z$ operators. Thus, a collection of CZ edges between the two blocks is specified by a bilinear pairing
\begin{equation}
    \zeta_{AB} : C_1^A \times C_1^B \to \mathbb F_2.
\end{equation}
Using the fixed computational bases, we identify each qubit space with its dual, such that this pairing is represented by a matrix $\zeta_1 \in \mathbb F_2^{n_A \times n_B}$, equivalently a linear map $\zeta_1 : (C_1^B)^* \to C_1^A$, where $(\zeta_1)_{ab} = 1$ indicates that a CZ is applied from qubit $a$ to qubit $b$. Here we refer to the dual-reversed complex 
\begin{equation}
    B_\bullet^\vee :
    (C_0^B)^*
    \xrightarrow{(\partial_1^B)^{\mathsf T}}
    (C_1^B)^*
    \xrightarrow{(\partial_2^B)^{\mathsf T}}
    (C_2^B)^* ,
\end{equation}
with \(
B_2^\vee=(C_0^B)^*,
 B_1^\vee=(C_1^B)^*\), and \(B_0^\vee=(C_2^B)^* 
\), to define the space of diagonal chain maps $ \boldsymbol{\zeta}:B_\bullet^\vee\to A_\bullet$.

\begin{definition}
    The space of diagonal chain maps from $B$ to $A$ is 
    \begin{equation}
    \operatorname{Diag}_{\mathrm{Ch}}(B,A) := \operatorname{Hom}_{\mathrm{Ch}}(B_\bullet^\vee,A_\bullet).
    \end{equation}
\end{definition}
\hspace{-6mm}Thus, an element $\zeta_\bullet = (\zeta_2, \zeta_1, \zeta_0) \in \operatorname{Diag}_{\mathrm{Ch}}(B, A)$ consists of maps
\begin{equation}
    \zeta_2 : (C_0^B)^* \to C_2^A,
    \qquad
    \zeta_1 : (C_1^B)^* \to C_1^A,
    \qquad
    \zeta_0 : (C_2^B)^* \to C_0^A,
\end{equation}
satisfying
\begin{equation}
    \partial_2^A\zeta_2
    =
    \zeta_1(\partial_1^B)^{\mathsf T},
    \qquad
    \partial_1^A\zeta_1
    =
    \zeta_0(\partial_2^B)^{\mathsf T}.
    \label{eq:chain-map-eqns-dual}
\end{equation}
We are primarily interested in the middle component $\zeta_1$, since it
specifies the physical CZ pattern. We define the space of degree-$1$ maps that admit a completion to a diagonal chain map as 
\begin{equation}
    \operatorname{Diag}_1(B,A)
    :=
    \left\{
    \zeta_1\in\mathbb F_2^{n_A\times n_B}
    \;\middle|\;
    \exists\,\zeta_2,\zeta_0
    \text{ such that }
    (\zeta_2,\zeta_1,\zeta_0)
    \in
    \operatorname{Diag}_{\mathrm{Ch}}(B,A)
    \right\}.
    \label{eq:diag1-def}
\end{equation}
In terms of the parity-check matrices of codes $A$ and $B$,  Eqs.~\eqref{eq:chain-map-eqns-dual} become
\begin{equation}
    (H_Z^A)^{\mathsf T}\zeta_2 = \zeta_1(H_X^B)^{\mathsf T}, \qquad H_X^A\zeta_1 = \zeta_0 H_Z^B.
    \label{eq:chain-map-eqns-dual-explicit}
\end{equation}
These equations encode stabilizer-preservation conditions for the inter-code CZ pattern $\zeta_1$. An $X$-stabilizer of $B$ picks up a $Z$ operator on $A$, and the first equation requires this operator to be a $Z$ stabilizer of $A$. Similarly, an $X$-stabilizer of $A$ picks up a $Z$ operator on $B$, and the second equation requires this operator be a $Z$ stabilizer on $B$. Since CZ gates are diagonal, all $Z$ stabilizers are fixed. The following proposition gives the diagonal analogue of Proposition~\ref{prop:hom1}.

\begin{proposition}
A linear map $\zeta_1 : (C_1^B)^* \to C_1^A$ lies in
$\operatorname{Diag}_1(B,A)$ if and only if
\begin{align}
\zeta_1\!\left(\im(\partial_1^B)^{\mathsf T}\right)
&\subseteq \im\partial_2^A,
\label{eq:diag1-cond1}\\
\zeta_1\!\left(\ker(\partial_2^B)^{\mathsf T}\right)
&\subseteq \ker\partial_1^A.
\label{eq:diag1-cond2}
\end{align}
\label{prop:diag1-chain-map-eqs}
\end{proposition}
Finally, we can define the logical CZ action. Given the Pauli propagation under CZ gates, the relevant logical data are the $X$ logical classes, represented by cohomology 
\begin{equation}
    H^1
:=
\ker(\partial_2^{\mathsf T})/
\operatorname{im}(\partial_1^{\mathsf T}).
\end{equation}
\begin{lemma}
    Any $\zeta_1 \in\operatorname{Diag}_1(B,A)$ induces a well-defined bilinear pairing 
    \begin{equation}
        \zeta_\star : H^1(A)\times H^1(B) \to \mathbb F_2,
        \qquad
        ([x_A],[x_B])
    \longmapsto
    x_A^{\mathsf T}\zeta_1x_B.
    \end{equation}
    \label{lemma:diag-logical-pairing}
\end{lemma}
The proofs of Proposition~\ref{prop:diag1-chain-map-eqs} and
Lemma~\ref{lemma:diag-logical-pairing} are direct dual analogues of
Proposition~\ref{prop:hom1} and Lemma~\ref{lemma:induced_op}. In particular, the two inclusions in Proposition~\ref{prop:diag1-chain-map-eqs} are precisely the conditions that allow $\zeta_1(\partial_1^B)^{\mathsf T}$ to factor through $\partial_2^A$ and $\partial_1^A\zeta_1$ to factor through $(\partial_2^B)^{\mathsf T}$. For the pairing, changing $x_B$ by an element of
$\operatorname{im}(\partial_1^B)^{\mathsf T}$ changes $\zeta_1x_B$ by a
boundary in $A$, which pairs trivially with $x_A\in\ker(\partial_2^A)^{\mathsf T}$; changing $x_A$ by an element of
$\operatorname{im}(\partial_1^A)^{\mathsf T}$ gives $u^{\mathsf T}\partial_1^A\zeta_1x_B=0$ since $\zeta_1x_B\in\ker\partial_1^A$.

\subsection{Vectorised construction of \texorpdfstring{$\operatorname{Diag}_1(B,A)$}{Diag1(B,A)}}

As in Sec.~\ref{ssec:vectorisation}, Proposition~\ref{prop:diag1-chain-map-eqs} gives a direct kernel description of admissible degree-$1$ CZ couplings. Let $N_B^\vee$ be a matrix whose columns form a basis of $\ker(\partial_2^B)^{\mathsf T}$, and let $U_A$ be a matrix whose columns form a basis of $\ker\bigl((\partial_2^A)^{\mathsf T}\bigr)$. Then, the conditions in Eqs.~\eqref{eq:diag1-cond1} and~\eqref{eq:diag1-cond2} are equivalent to 
\begin{align}
\partial_1^A \zeta_1 N_B^\vee &= 0,
\label{eq:diag1-condition-kernel}\\
U_A^{\mathsf T}\zeta_1(\partial_1^B)^{\mathsf T} &= 0.
\label{eq:diag1-condition-image}
\end{align}
Using the same column-stacking convention and vectorisation identity as in the main text, Eqs.~\eqref{eq:diag1-condition-kernel} and
\eqref{eq:diag1-condition-image} give rise to the linear system $P_{\mathrm{CZ}}\operatorname{vec}(\zeta_1)=0$, where
\begin{equation}
    P_{\mathrm{CZ}}
    :=
    \begin{pmatrix}
    (N_B^\vee)^{\mathsf T}\otimes \partial_1^A \\
    \partial_1^B\otimes U_A^{\mathsf T}
    \end{pmatrix}.
    \label{eq:diag1-linear-system}
\end{equation}
The following analogue ensues.
\begin{corollary}
Under fixed bases and the vectorisation convention, the restricted vectorisation map
\begin{equation}
\operatorname{vec}\big|_{\operatorname{Diag}_1(B,A)}
:
\operatorname{Diag}_1(B,A)
\xrightarrow{\cong}
\ker P_{\mathrm{CZ}}
\end{equation}
is a linear isomorphism.
\end{corollary}

\subsection{Induced logical CZ action}

A solution $\zeta_1$ specifies a physical CZ circuit. For every nonzero entry $(\zeta_1)_{ab} = 1$, we apply a CZ gate between qubit $a$ of code $A$ and $b$ of code $B$. Since all CZ gates commute, the circuit is again represented by a bipartite interaction graph, with depth and weight defined as in Sec.~\ref{sec:cp-sat}.

By Lemma~\ref{lemma:diag-logical-pairing}, $\zeta_1$ induces a well-defined bilinear pairing $\zeta_\star$. Let us choose matrices $\boldsymbol L_X^A\in\mathbb F_2^{k_A\times n_A}$ and $\boldsymbol L_X^B\in\mathbb F_2^{k_B\times n_B}$ whose rows are representatives of bases of $H^1(A)$ and $H^1(B)$, respectively. The matrix representation of the induced logical CZ pairing in these bases is 
\begin{equation}
    \Gamma_{AB}(\zeta_1) := 
    \boldsymbol L_X^A\zeta_1(\boldsymbol L_X^B)^{\mathsf T}
\in\mathbb F_2^{k_A\times k_B}.
\label{eq:induced_diag_Gamma_AB}
\end{equation}
An entry in $(\Gamma_{AB})_{ij}$ indicates whether the physical CZ pattern induces a logical CZ coupling between the $i^{\mathrm{th}}$ logical qubit of code $A$ and the $j^{\mathrm{th}}$ logical qubit of code $B$. Equivalently, the CZ circuit maps the Pauli logical operators following 
\begin{equation}
    \overline X_i^A
    \mapsto
    \left(\overline X_i^A
    \prod_{j=1}^{k_B}
    (\overline Z_j^B)^{(\Gamma_{AB})_{ij}} \right),
    \qquad
    \overline X_j^B
    \mapsto
    \left(
    \prod_{i=1}^{k_A}
    (\overline Z_i^A)^{(\Gamma_{AB})_{ij}}
    \right)
    \overline X_j^B,
    \label{eq:cz-logical-pauli-action}
\end{equation}
up to stabilizers, while all logical $Z$ operators are fixed.

Since $\Gamma_{AB}$ depends linearly on $\zeta_1$, any linear constraint on the induced logical CZ action defines an affine restriction of $\operatorname{Diag}_1(B,A)$. Let 
\begin{equation}
\mathcal T_{{\mathrm{CZ}}}:
\mathbb F_2^{k_A\times k_B}\to \mathbb F_2^t
\end{equation}
be a linear map encoding the desired logical constraints, and let $b \in \mathbb F_2^t$. We define the corresponding logically targeted
diagonal family by
\begin{equation}
\operatorname{Diag}_1(B,A)\big|_{(\mathcal T_{\mathrm{CZ}},b)}
:=
\left\{
\zeta_1\in\operatorname{Diag}_1(B,A)
\;\middle|\;
\mathcal T_{\mathrm{CZ}}\!\left(\Gamma_{AB}(\zeta_1)\right)=b
\right\}.
\label{eq:targeted-diag-family-general}
\end{equation}
Since $\mathcal T_{\mathrm{CZ}}\circ\Gamma_{AB}$ is linear, this restricted solution space is
either empty or an affine subspace of $\operatorname{Diag}_1(B,A)$. In the implementation used here, we take the special case in which
$\mathcal T_{\mathrm{CZ}}$ is the vectorisation map and
$b=\operatorname{vec}(\Gamma_{AB}^{\mathrm{tgt}})$. When the affine system is feasible, its solution may be written as 
\begin{equation}
\zeta_1(y)
=
\zeta_1^{(0)}
+
\sum_{m=1}^r y_m\widetilde\Delta_m,
\qquad
y_m\in\mathbb F_2,
\label{eq:affine-cz-fam}
\end{equation}
where $\zeta_1^{(0)}$ is one solution satisfying the targeted logical constraints, and the matrices $\widetilde\Delta_m$ span the homogeneous kernel of the combined stabilizer-preservation and logical-targeting constraints.

As in the case for homomorphic CNOT solutions, the affine family in Eq.~\eqref{eq:affine-cz-fam} can be passed to the CP-SAT optimisation layer. The binary variables $y_m$ select an element of the logically targeted affine space, while auxiliary Boolean variables encode the entries of $\zeta_1(y)$. One may then minimise a chosen objective such as depth, weight, or impose constraints and search for a feasible solution.

\section{Algorithms}

We summarise the algorithmic components underlying the construction of a physical CNOT circuit implementing a prescribed
logical action between CSS codes $A$ and $B$. The procedure consists of four stages. Firstly, a basis of the degree-$1$ hom-space $\HomI{B}{A}$ is constructed from the kernel description derived in
Sec.~\ref{sec:build-hom-space} (Algorithm~\ref{alg:hom1-family}). Secondly, the induced logical action of a candidate coupling is extracted in terms of the matrices
$\Gamma_Z$ and $\Gamma_X$ (Algorithm~\ref{alg:extract-gammaxz}). Third, logical targeting is imposed by
restricting the hom-space to an affine family satisfying the desired logical constraints (Algorithm~\ref{alg:restrict-by-gammaxz}). Finally, optional sparse basis replacement (Algorithm~\ref{alg:sparse-basis}) and CP-SAT optimisation are used to search this family for shallow and
sparse implementations. We note that Algorithm~\ref{alg:extract-gammaxz} is written for symplectically normalised logical representatives satisfying $L_X^A(L_Z^A)^{\mathsf T}=I$ and $L_X^B(L_Z^B)^{\mathsf T}=I$, such that the coefficients of the induced logical action can be
extracted directly from logical commutation pairings, and
Proposition~\ref{prop:Gamma-compat} reduces to $\Gamma_X =\Gamma_Z^{\mathsf T}$. For general logical bases, the same logical action may instead be
computed from Eqs.~\eqref{eq:GammaZ-def}
and~\eqref{eq:GammaX-def}, or equivalently by using the pairing
matrices $\Omega_A$ and $\Omega_B$ from
Proposition~\ref{prop:Gamma-compat}. \\

\begin{algorithm}[H]
\caption{Construction of $\HomI{B}{A}$}
\label{alg:hom1-family}

\KwIn{
CSS chain complexes
\[
A_\bullet: C_2^A \xrightarrow{\partial_2^A} C_1^A \xrightarrow{\partial_1^A} C_0^A,
\qquad
B_\bullet: C_2^B \xrightarrow{\partial_2^B} C_1^B \xrightarrow{\partial_1^B} C_0^B
\]
}
\KwOut{
A basis $\{\Delta_i\}_{i=1}^q$ of $\HomI{B}{A}$, so that every admissible middle map has a unique expansion
\[
\gamma_1=\sum_{i=1}^q x_i\Delta_i,\qquad x_i\in\mathbb F_2.
\]
}

$n_A \leftarrow \dim C_1^A$, \quad $n_B \leftarrow \dim C_1^B$\;

Compute a matrix $N_B$ whose columns form a basis of $\ker\partial_1^B$\;

Compute a matrix $U_A$ whose columns form a basis of $\ker(\partial_2^A)^{\mathsf T}$\;

Construct the constraint blocks
\[
P_1 \leftarrow N_B^{\mathsf T}\otimes \partial_1^A,
\qquad
P_2 \leftarrow (\partial_2^B)^{\mathsf T}\otimes U_A^{\mathsf T}
\]
corresponding respectively to the conditions
\[
\partial_1^A\gamma_1N_B=0,
\qquad
U_A^{\mathsf T}\gamma_1\partial_2^B=0.
\]

Set
\[
P \leftarrow
\begin{pmatrix}
P_1\\
P_2
\end{pmatrix}.
\]

Compute a basis $\{\delta_1,\dots,\delta_q\}$ of $\ker P$\;

\For{$i=1,\dots,q$}{
Set
\[
\Delta_i \leftarrow \operatorname{vec}^{-1}(\delta_i)
\in \mathbb F_2^{n_A\times n_B}
\]
using the fixed column-stacking convention\;
}

\Return{$\{\Delta_i\}_{i=1}^q$}
\end{algorithm}

\begin{algorithm}[H]
\caption{Extraction of the induced logical action $(\Gamma_X, \Gamma_Z)$}
\label{alg:extract-gammaxz}
\KwIn{
A coupling matrix $\gamma_1\in\mathbb{F}_2^{n_A\times n_B}$, $H_X^A, H_Z^A, H_X^B, H_Z^B$, and a symplectic logical representative basis $\boldsymbol L_X^A,\boldsymbol L_Z^A,\boldsymbol L_X^B,\boldsymbol L_Z^B$ 
}
\KwOut{$\Gamma_Z(\gamma_1)\in\mathbb F_2^{k_A\times k_B}$, $\Gamma_X(\gamma_1)\in\mathbb F_2^{k_B\times k_A}$
}
Initialise $\Gamma_Z\leftarrow 0_{k_A\times k_B}$ and $\Gamma_X\leftarrow 0_{k_B\times k_A}$\;
\For{$j=1, \dots, k_B$}
{
Let $z_B^{(j)}$ be the $j^{\mathrm{th}}$ row of $\boldsymbol L_Z^B$\;
Compute 
\[
z_A \leftarrow \gamma_1 (z_B^{(j)})^{\mathsf T}
\]
Verify that $z_A\in \ker H_X^A$\;
Set
\[
\Gamma_Z(\gamma_1)e_j \leftarrow \boldsymbol L_X^A z_A.
\]
}
\For{$i=1,\dots,k_A$}{
Let $x_A^{(i)}$ be the $i^{\mathrm{th}}$ row of $\boldsymbol L_X^A$\;

Compute 
\[
x_B \leftarrow \gamma_1^{\mathsf T}(x_A^{(i)})^{\mathsf T}
\]\;

Verify that $H_X^A z_A=0$ and $H_Z^B x_B=0$\;

Set 
\[
\Gamma_X(\gamma_1)e_i \leftarrow \boldsymbol L_Z^B x_B.
\]

}
\Return{$(\Gamma_Z, \Gamma_X)$}
\end{algorithm}

\begin{algorithm}[H]
\caption{Restriction of a degree-1 family by a target logical-$Z$ action}
\label{alg:restrict-by-gammaxz}

\KwIn{
$H_X^A,H_Z^A,H_X^B,H_Z^B$,
$\boldsymbol L_X^A, \boldsymbol L_Z^A, \boldsymbol L_X^B, \boldsymbol L_Z^B$, 
$\Gamma_Z^{\mathrm{tgt}}$, and an affine family
\[
\mathcal{A} =
\left\{
\gamma_1=\gamma_1^{(0)}+\sum_{i=1}^q x_i\Delta_i
\;\middle|\;
x_i\in\mathbb F_2
\right\} \subseteq \HomI{B}{A}.
\]
}
\KwOut{Either $\varnothing$, or a restricted affine subfamily $\mathcal{A}'$ whose elements realise $\Gamma_Z^{\mathrm{tgt}}$.}

Compute
\[
\Gamma_Z^{(0)} := \Gamma_Z(\gamma_1^{(0)}).
\]

\For{$i=1,\dots,q$}{
Compute
\[
\Gamma_Z^{(i)} := \Gamma_Z(\Delta_i).
\]
Set
\[
\lambda_i := \operatorname{vec}(\Gamma_Z^{(i)}).
\]
}

Set
\[
b := \operatorname{vec}\!\left(
\Gamma_Z^{\mathrm{tgt}}+\Gamma_Z^{(0)}
\right).
\]

Set
\[
\Lambda :=
\begin{pmatrix}
\lambda_1 & \lambda_2 & \cdots & \lambda_q
\end{pmatrix}.
\]

Solve
\[
\Lambda x=b
\]
over $\mathbb F_2$\;

\If{this system is infeasible}{
\Return $\varnothing$\;
}

Choose a solution $\bar{x}\in\mathbb F_2^q$, and let $\{\eta_1,\dots,\eta_m\}$ be a basis of $\ker \Lambda$\;

Set
\[
\gamma_1^\star := \gamma_1^{(0)}+\sum_{i=1}^q \bar{x}_i\Delta_i.
\]

\For{$j=1,\dots,m$}{
Set
\[
\widetilde{\Delta}_j := \sum_{i=1}^q(\eta_j)_i\Delta_i.
\]
}

Set
\[
\mathcal{A}' \leftarrow
\left\{
\gamma_1=\gamma_1^\star+\sum_{j=1}^m y_j\widetilde{\Delta}_j
\;\middle|\;
y_j\in\mathbb F_2
\right\}.
\]

In symplectically normalised logical bases, $\Gamma_X(\gamma_1)=(\Gamma_Z^{\mathrm{tgt}})^{\mathsf T}$ for every $\gamma_1\in\mathcal A'$ by Proposition~\ref{prop:Gamma-compat}\;

\Return $\mathcal A'$\;
\end{algorithm}

\begin{algorithm}[H]
\caption{Sparse basis replacement for a restricted family}
\label{alg:sparse-basis}

\KwIn{
An affine family $\mathcal A$,
\[
\mathcal{A} =
\left\{
\gamma_1=\gamma_1^{(0)}+\sum_{i=1}^q x_i\Delta_i
\;\middle|\;
x_i\in\mathbb F_2
\right\} \subseteq \HomI{B}{A},
\]
with linear part
\[
V = \operatorname{span}\{\Delta_1,\dots,\Delta_q\},
\]
and parity-check matrices
\[
H_X^B \in \mathbb F_2^{r_X^B\times n_B},
\qquad
H_Z^A \in \mathbb F_2^{r_Z^A\times n_A}
\]
}
\KwOut{
Either the original family $\mathcal A$, or an equivalent reparametrisation $\mathcal S$ with sparse generators $\widetilde{\Delta}_j$
}

\[
\mathcal K \leftarrow \varnothing
\]

\For{$i=1,\dots,n_A$}{
    \For{$p=1,\dots,r_X^B$}{
        Let $h_p$ be the $p^{\mathrm{th}}$ row of $H_X^B$\;
        \[
        \mathcal K \leftarrow \mathcal K \cup \{\, e_i h_p \,\}
        \]
        where $e_i\in\mathbb F_2^{n_A}$ is a standard basis column vector.
    }
}

\For{$j=1,\dots,n_B$}{
    \For{$s=1,\dots,r_Z^A$}{
        Let $z_s$ be the $s^{\mathrm{th}}$ row of $H_Z^A$\;
        \[
        \mathcal K \leftarrow \mathcal K \cup \{\, z_s^{\mathsf T} e_j^{\mathsf T} \,\}
        \]
        where $e_j\in\mathbb F_2^{n_B}$ is a standard basis column vector.
    }
}

\[
\mathcal K' \leftarrow
\left\{
K\in\mathcal K
\;\middle|\;
\operatorname{vec}(K)\in
\operatorname{span}\bigl(
\operatorname{vec}(\Delta_1),\dots,\operatorname{vec}(\Delta_q)
\bigr)
\right\}
\]

Extract a linearly independent subset
\[
\{\widetilde{\Delta}_1,\dots,\widetilde{\Delta}_{q'}\}\subseteq \mathcal K'
\]
with
\[
\operatorname{span}\{\widetilde{\Delta}_1,\dots,\widetilde{\Delta}_{q'}\}
=
\operatorname{span}(\mathcal K')
\]

\eIf{$\operatorname{rank}\!\bigl[\operatorname{vec}(\widetilde{\Delta}_1),\dots,\operatorname{vec}(\widetilde{\Delta}_{q'})\bigr]
=
\operatorname{rank}\!\bigl[\operatorname{vec}(\Delta_1),\dots,\operatorname{vec}(\Delta_q)\bigr]$}{
    \[
    \mathcal S =
    \left\{
    \gamma_1=\gamma_1^{(0)}+\sum_{j=1}^{q'} y_j\widetilde{\Delta}_j
    \;\middle|\;
    y_j\in\mathbb F_2
    \right\}
    \]
    \Return{$\mathcal S$}
}{
    \Return{$\mathcal A$}
}
\end{algorithm}

\section{Example circuits}

Examples of small physical circuits mentioned in Table~\ref{tab:circ_table} to illustrate the solutions found using our method. In Fig.~\ref{fig:steane-to-surface-3} we show a depth-2 distance-preserving circuit from the Steane code to the distance-$3$ rotated surface code, realising a logical (homomorphic) $\overline{\mathrm{CNOT}}_{A_0\to B_0}$ between the logical qubits of the codes. Fig.~\ref{fig:15_1_3_to_surface_3} shows the same logical operation between the $[[15, 1, 3]]$ quantum Reed-Muller code and distance-$3$ surface code, also distance-preserving with depth 2. To interpret the circuits we show below the check matrices used to define the stabilizer generators of the three codes mentioned above. 
\begin{equation}
\scriptsize
\begin{gathered}
H_X^{\mathrm{Steane}}=H_Z^{\mathrm{Steane}}=
\begin{pmatrix}
1&1&1&1&0&0&0\\
0&1&1&0&1&1&0\\
1&1&0&0&1&0&1
\end{pmatrix},
\\[1em]
H_X^{\mathrm{surf}}=
\begin{pmatrix}
0&1&1&0&0&0&0&0&0\\
1&1&0&1&1&0&0&0&0\\
0&0&0&0&1&1&0&1&1\\
0&0&0&0&0&0&1&1&0
\end{pmatrix},
\qquad
H_Z^{\mathrm{surf}}=
\begin{pmatrix}
1&0&0&1&0&0&0&0&0\\
0&1&1&0&1&1&0&0&0\\
0&0&0&1&1&0&1&1&0\\
0&0&0&0&0&1&0&0&1
\end{pmatrix},
\\[1em]
H_X^{\mathrm{15RM}}=
\begin{pmatrix}
1&0&0&0&1&1&1&0&0&0&1&1&1&0&1\\
0&1&0&0&1&0&0&1&1&0&1&1&0&1&1\\
0&0&1&0&0&1&0&1&0&1&1&0&1&1&1\\
0&0&0&1&0&0&1&0&1&1&0&1&1&1&1
\end{pmatrix},
\\[1em]
H_Z^{\mathrm{15RM}}=
\begin{pmatrix}
1&0&0&0&1&1&1&0&0&0&1&1&1&0&1\\
0&1&0&0&1&0&0&1&1&0&1&1&0&1&1\\
0&0&1&0&0&1&0&1&0&1&1&0&1&1&1\\
0&0&0&1&0&0&1&0&1&1&0&1&1&1&1\\
0&0&0&0&1&0&0&0&0&0&1&1&0&0&1\\
0&0&0&0&0&1&0&0&0&0&1&0&1&0&1\\
0&0&0&0&0&0&1&0&0&0&0&1&1&0&1\\
0&0&0&0&0&0&0&1&0&0&1&0&0&1&1\\
0&0&0&0&0&0&0&0&1&0&0&1&0&1&1\\
0&0&0&0&0&0&0&0&0&1&0&0&1&1&1
\end{pmatrix}.
\end{gathered}
\end{equation}

\begin{figure}[h]
    \centering
    \includegraphics[width=0.6\linewidth]{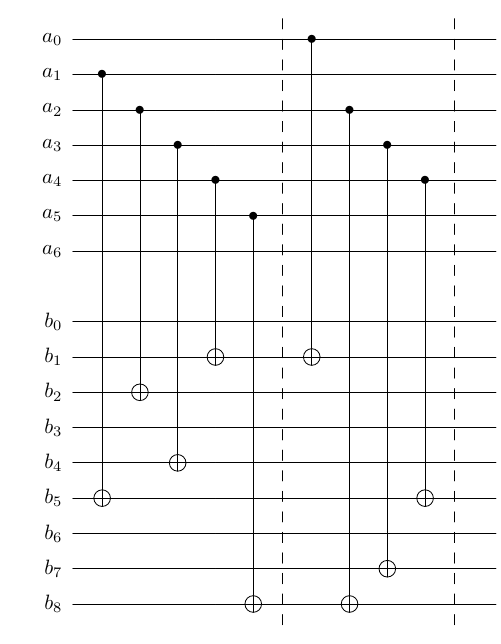}
    \caption{\small \textbf{Example physical circuit implementing a homomorphic CNOT via a chain  map.} Distance-preserving depth-2 homomorphic CNOT circuit from the Steane code to the distance-$3$ rotated surface code. Wires labeled $a_i$ and $b_i$ represent the physical qubits of these two codes, respectively. The dashed lines separate the time steps.}
    \label{fig:steane-to-surface-3}
\end{figure}

\begin{figure}[h]
    \centering
    \includegraphics[width=0.6\linewidth]{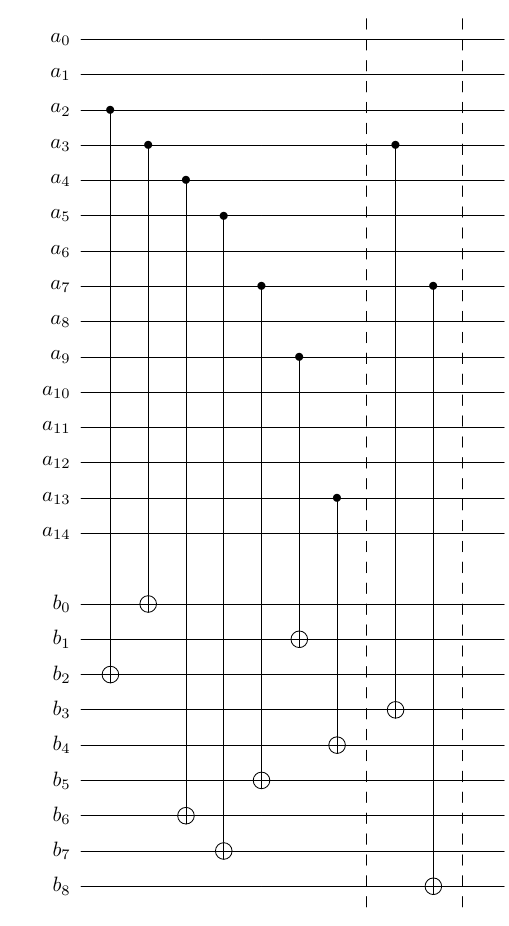}
    \caption{\small  \textbf{Example physical circuit implementing a homomorphic CNOT via a chain  map.} Distance-preserving depth-2 homomorphic CNOT circuit from the $[[15,1,3]]$ quantum Reed-Muller code to the distance-$3$ rotated surface code. Wires labeled $a_i$ and $b_i$ represent the physical qubits of these two codes, respectively. The dashed lines separate the time steps.}
    \label{fig:15_1_3_to_surface_3}
\end{figure}



\end{document}